\documentclass[a4paper,USenglish, thm-restate]{lipics-v2021}
\usepackage[utf8]{inputenc}
\usepackage{todonotes}
\usepackage{amsmath, amsthm, amssymb}
\usepackage{graphicx}
\usepackage{algorithm, algorithmicx}
\usepackage[noend]{algpseudocode}
\usepackage{xspace}
\usepackage{microtype}
\usepackage{thmtools, thm-restate}
\hideLIPIcs

\newif\ifmergedfigs

\graphicspath{{figures/}}%helpful if your graphic files are in another director

\title{Unlabeled Multi-Robot Motion Planning with Tighter Separation Bounds}

\author{Bahareh Banyassady}{Freie Universit\"at Berlin, Germany}{bahareh.banyassady@fu-berlin.de}
{0000-0002-3422-9028}%{ORCID}
{}%{(Optional) author-specific funding acknowledgements}

\author{Mark de Berg}{TU Eindhoven, the Netherlands}{m.t.d.berg@tue.nl}
{0000-0001-5770-3784}%{ORCID}
{}%{(Optional) author-specific funding acknowledgements}

\author{Karl Bringmann}{Saarland University and Max Planck Institute for Informatics, Germany}{bringmann@cs.uni-saarland.de}
{}%{ORCID}
{}%{(Optional) author-specific funding acknowledgements}

\author{Kevin Buchin}{TU Dortmund, Germany}{kevin.buchin@tu-dortmund.de}
{000-0002-3022-7877}%{ORCID}
{}%{(Optional) author-specific funding acknowledgements}

\author{Henning Fernau}{University of Trier, Germany}{fernau@uni-trier.de}
{0000-0002-4444-3220}%{ORCID}
{}%{(Optional) author-specific funding acknowledgements}

\author{Dan Halperin}{Tel Aviv University, Israel}{danha@post.tau.ac.il}
{0000-0002-3345-3765}%{ORCID}
{}%{(Optional) author-specific funding acknowledgements}

\author{Irina Kostitsyna}{TU Eindhoven, the Netherlands}{i.kostitsyna@tue.nl}
{0000-0003-0544-2257}%{ORCID}
{}%{(Optional) author-specific funding acknowledgements}

\author{Yoshio Okamoto}{The University of Electro-Communications, Japan}{okamotoy@uec.ac.jp}
{0000-0002-9826-7074}%{ORCID}
{}%{(Optional) author-specific funding acknowledgements}

\author{Stijn Slot}{Adyen, the Netherlands}{slot.stijn@gmail.com}
{}%{ORCID}
{}%{(Optional) author-specific funding acknowledgements}

\authorrunning{B. Banyassady et al.}
 %TODO mandatory: Please choose ACM 2012 classifications from https://dl.acm.org/ccs/ccs_flat.cfm 
 %TODO mandatory; please add comma-separated list of keywords

\category{} %optional, e.g. invited paper

\relatedversion{A shorter version of this paper appeared in the Proceedings of the 38th International Symposium on Computational Geometry (SoCG 2022).} %optional, e.g. full version hosted on arXiv, HAL, or other respository/website
\nolinenumbers %uncomment to disable line numbering

\let\emptyset\varnothing

\newcommand{\R}{\mathbb{R}}
\newcommand{\W}{\mathcal{W}}
\newcommand{\F}{\mathcal{F}}
\newcommand{\D}{\mathcal{D}}

\newcommand{\weight}{charge\xspace}
\newcommand{\IS}{\mIS_{\rightarrow}}
\newcommand{\mIS}{\mathcal{I}}
\newcommand{\Fminus}{F^-}
\newcommand{\Dminus}{\D^-}

\begin{document}
	
	\maketitle

\begin{abstract}
We consider the unlabeled motion-planning problem of $m$ unit-disc robots moving in a simple polygonal workspace of $n$ edges.
The goal is to find a motion plan that moves the robots to a given set of $m$ target positions.
For the unlabeled variant, it does not matter which robot reaches which target position as long as all target positions are occupied in the end.

If the workspace has narrow passages such that the robots cannot fit through them, then the free configuration space, representing all possible unobstructed positions of the robots, %\danny{I'd say more precisely ``then the free (configuration) space representing all the free positions of a single disc robot''}
will consist of multiple connected components.
Even if in each component of the free space the number of targets matches the number of start positions, the motion-planning problem does not always have a solution when the robots and their targets are positioned very densely.
In this paper, we prove tight bounds on how much separation between start and target positions is necessary to always guarantee a solution. 
Moreover, we describe an algorithm that always finds a solution in time $O(n \log n + mn + m^2)$ if the separation bounds are met. 
Specifically, we prove that the following separation is sufficient: any two start positions are 
at least distance~$4$ apart, any two target positions are at least distance~4 apart, and 
any pair of a start and a target positions is at least distance~$3$ apart. We further show that 
when the free space consists of a single connected component, the separation between start and target positions is not necessary.
\end{abstract}

% ================================================
% ================= Introduction =================
% ================================================

\section{Introduction}
\label{sec:introduction}

% new text, Danny 16/4/21

Multi-robot systems are already playing a central role in manufacturing, warehouse logistics, inspection of large structures (e.g., bridges), monitoring of natural resources, and in the future they are expected to expand to other domains such as space exploration, search-and-rescue tasks and more.
One of the key ingredients necessary for endowing multi-robot systems with autonomy 
%\rev{In order to reach more “autonomy”, there are other issues to overcome as well, e.g., game-theoretic/multi-agent ones. Maybe mention that the algorithms you present in your paper rely on some central system to compute motion plans instead of having robots act autonomously.}
is the ability to plan collision-free motion paths for its constituent robots towards desired target positions.

% MRMP, C-space
In the basic multi-robot motion-planning (MRMP) problem several robots are operating in a common environment.
We are given a set of start positions and a set of desired target positions for these robots,
and we wish to compute motions that will bring the robots to the targets 
% positions (kevin: shortened to get on 1 page)
while avoiding collisions with obstacles and the other robots.
We distinguish between two variants of MRMP, \emph{labeled} and \emph{unlabeled}, depending on whether each robot has to reach a specific target.
%Furthermore, the most common variant is 
In \emph{labeled} robot motion planning, each robot has a designated target position. % that needs to reach.
In contrast, in the \emph{unlabeled} variant%~\cite{kloder2006path,turpin2013concurrent}
, which we study here, each robot only needs to reach \emph{some} target position;
it does not matter which robot reaches which target as long as at the end each target position is occupied by a robot.

MRMP is an extension of the extensively studied \emph{single} robot motion-planning problem (see, e.g., \cite{choset2005principles,halperin2018algorithmic,lavalle2006planning}).
The multi-robot case is considerably harder~\cite{hearn2005pspace,hopcroft1984complexity,spirakis1984strong}, since the dimension of the \emph{configuration space} grows with the number of robots in the system.
The configuration space of a robot system is a parametric representation of all the possible configurations of the system, which are determined by specifying a real value for each independent parameter (degree of freedom) of the system. 
%Hence, the dimension of the C-space is the overall number of degrees of freedom of the system. 
%For multi-robot problems, %that we study in this paper, 
%the dimension of the C-space is the total sum of degrees of freedom of the individual robots. 

The system we study consists of unit-disc robots moving in the plane; see below for a more formal problem statement. 
Not only is this a reasonably faithful representation of existing robot systems (e.g., in logistics), but it already encapsulates the essential hardships of MRMP, as MRMP for planar systems with simply-shaped robots are known to be hard~\cite{hopcroft1984complexity,solovey2016hardness}.
Surprisingly, when we assume some minimum spacing between the start and target positions, the problem for robots moving in a simple polygon always has a solution, and the solution can be found in polynomial time, as shown by Adler et al.~\cite{adler2015efficient}.
The \emph{separation}, the minimum distance between the start and target positions, thus plays a key role in the difficulty of the problem.
However the separation bounds assumed by Adler et al. are not proven to be tight, so the question remains for what separation bounds the problem is always solvable.
In this paper we determine the minimal separation needed to ensure that the motion-planning problem has a solution, improving on the bounds obtained by Adler et al.
We also describe an algorithm that plans such motions efficiently, relying on the new bounds that we obtain.

\subparagraph{Related work.}
The multi-robot motion planning problem has received much attention over the years.
Already in 1983, the problem was described in a paper on the \emph{Piano Mover's problem} by Schwartz and Sharir~\cite{schwartz1983piano}.
Later that year, an algorithm for the case of two or three disc robots moving in a polygonal environment with $n$ polygon vertices was described, running in $O(n^3)$ and $O(n^{13})$ time respectively~\cite{schwartz1983piano2}.
This was later improved by Yap~\cite{yap1984coordinating} to $O(n^2)$ and $O(n^3)$ for two and three robots respectively, using the \emph{retraction method}.
A general approach using \emph{cell decomposition} was later developed in 1991 by Sharir and Sifrony~\cite{sharir1991coordinated} that could deal with a variety of robot pairs in $O(n^2)$ time.

% Complexity labeled
Unfortunately, when the number of robots increases beyond a fixed constant, the problem becomes hard.
In 1984, a \emph{labeled} case of the multi-robot motion planning with disc robots and a simple polygonal workspace was shown to be strongly NP-hard~\cite{spirakis1984strong}. 
This is a somewhat weaker result than the PSPACE-hardness for many other motion planning problems. 
For rectangular robots in a rectangular workspace, however, the problem was shown to be PSPACE-hard~\cite{hopcroft1984complexity}.
This result has later been refined to show that for PSPACE-completeness it is sufficient to have only $1\times 2$ or $2\times 1$ robots in a rectangular workspace~\cite{hearn2005pspace}.

% Probabilistic roadmaps
The hardness results for the general problem, as well as the often complex algorithms that solve the problem exactly~\cite{halperin2018algorithmic}, led to the development of more practical solutions, which often trade completeness of the solution for simplicity and speed, and can successfully cope with motion-planning problems with many degrees of freedom. Most notable among the practical solutions are
\emph{sampling-based} (SB) techniques.  These include the celebrated Probabilistic Roadmaps (PRM)~\cite{kavraki1996probabilistic}, the Rapidly Exploring Random Trees (RRT)~\cite{kuffner2000rrtcpnnect}, and their numerous variants~\cite{choset2005principles,halperin2018robotics,lavalle2006planning}.
The probabilistic roadmaps can be widely applied to explore the high-dimensional configuration space, such as settings with a large number of robots or robots with many degrees of freedom.
However, in experiments by Sanchez and Latombe~\cite{sanchez2002using} already for 6 robots with a total of 36 degrees of freedom the algorithm requires prohibitively long time to find a solution. 
%\danny{check if this does not already use overmars-svestka}
% checked, it doesn't
Svestka and Overmars~\cite{svestka1998coordinated} suggested an SB algorithm specially tailored to many robots. Their solution still requires exorbitantly large roadmaps and is restricted to a small number of robots.  Solovey et al.~\cite{solovey2016finding} devised a more economical approach, dRRT (for discrete RRT), which is capable of coping with a larger number of robots, and which was extended to produce asymptotically optimal solutions~\cite{shome2020drrtstar}, namely converging to optimal (e.g., shortest overall distance) solution as the number of samples tends to infinity.

%Thus, for large number of robots with high degrees of freedom, such centralized, coupled algorithms are not sufficiently scalable even when using a sampling-based approach.

% Coupled
%Decoupled algorithms, where robots are first considered individually and issues are resolved locally, provide a scalable solution but at the cost of theoretical properties such as optimality and completeness.
%Therefore, many coupled algorithms have been proposed that combine individual (probabilistic) roadmaps for robots in a way that remains scalable yet keeps certain theoretical guarantees~\cite{gharbi2009roadmap, gravot2003method}.
%In particular, Dobson et al.~\cite{dobson2017scalable} show an algorithm called dRRT* that builds a roadmap for each robot and then implicitly searches the tensor product of these roadmaps in the composite space.
%They show that dRRT* is asymptotically-optimal, meaning the probability of finding the optimal solution asymptotically increases to~1 when the sampling size increases.

Regarding separability bounds, Solomon and Halperin~\cite{solomon2018motion} studied the labeled version of the unit-disc problem among polygonal obstacles in the plane, and showed that a solution always exists under a more relaxed monochromatic separation: each start or target position has an \emph{aura}, namely it resides inside a not-necessarily-concentric disc of radius 2, and the auras 
%\karl{The notion of aura has not been introduced yet} 
of two start positions (each being start or target) may overlap, as long as the aura of one robot does not intersect the other robot. They do not however make the distinction between monochromatic and bichromatic separation, and impose the same conditions for bichromatic auras as well.

% Complexity unlabeled
With respect to unlabeled motion planning, the problem was first considered by Kloder and Hutchinson~\cite{kloder2006path} in 2006.
In their paper they provide a sampling-based algorithm which is able to solve the problem.
In 2016, Solovey and Halperin~\cite{solovey2016hardness} have shown that for unit square robots the problem is PSPACE-hard using a reduction from \emph{non-deterministic constraint logic} (NCL)~\cite{hearn2005pspace}.
This PSPACE-hardness result also extends to the labeled variant for unit square robots.
Just recently, the unlabeled variant for two classes of disc robots with different radii was also shown to be PSPACE-hard~\cite{brocken2020multi}, with a similar reduction from NCL.
In the reduction the authors use robots of radius $\frac{1}{2}$ and $1$.
In contrast, the earlier NP-hardness result for disc robots by Spirakis and Yap~\cite{spirakis1984strong} required discs of many sizes with large differences in radii.

% Unlabeled algorithms
Fortunately, an efficient (polynomial-time) algorithm can still exist when some additional assumptions are made on the problem.
Turpin, Michael, and Kumar~\cite{turpin2013concurrent} consider a variant of the unlabeled motion-planning problem where the collection of free positions surrounding every start or target position is star-shaped.
This allows them to create an efficient algorithm for which the path-length is minimized.
In the paper by Adler et al.~\cite{adler2015efficient}, an $O(n \log n + mn + m^2)$ algorithm is given for the unlabeled variant, assuming the workspace is a simple polygon and the start and target positions are \emph{well-separated}, which is defined as minimum distance of four between any start or target position. 
Their algorithm is based on creating a motion graph on the start and target positions and then treating this as an \emph{unlabeled pebble game}, which can be solved in $O(S^2)$ where $S$ is the number of pebbles~\cite{kornhauser1984coordinating}.
Furthermore, in the paper by Adler et al.~\cite{adler2015efficient} the separation bound $4\sqrt{2} - 2$ ($\approx 3.646$) is shown to be sometimes necessary for the problem to always have a solution.
When the workspace contains obstacles, Solovey et al.~\cite{solovey2015motion} describe an approximation algorithm which is guaranteed to find a solution when one exists, assuming also that the start and target positions are \emph{well-separated} and a minimum distance of $\sqrt{5}$ between a start or target position and an obstacle.

Finally, we mention that multi-pebble motion on graphs, already brought up above, is part of a large body of work on motion planning in discrete domains, sometimes called multi-agent path finding (MAPF), and often adapted to solving continuous problems; see~\cite{stern2019multiagent} for a review, and
\cite{demaine2019coordinated,wagner2015subdimensional,yu2018constant,yu2016optimal}
for a sample of recent results.
%\rev{Regarding [4], it may be interesting to note that (monochromatic) separation plays an important role for unit disc robots in that paper.}

\subparagraph{Contributions.}
%\irina{I revised the contributions. In particular, I tried to remove or simplify descriptions that were too technical and detailed to understand at this point.}
%\mdb{Looks good to me.}
We  distinguish between two types of separability bounds: \emph{monochromatic}, denoted by~$\mu$, the separation between two start positions or between two target positions, and \emph{bichromatic}, denoted by~$\beta$, between a start and a target positions (see Figure~\ref{fig:definitions_example}).

After introducing necessary definitions and notation in Section~\ref{sec:definitions}, we begin with a lower bound construction for the monochromatic and bichromatic separation in Section~\ref{sec:separation-bounds}.
We prove that for $\mu= 4-\varepsilon$ or for $\beta= 3-\varepsilon$ (for arbitrarily small positive $\varepsilon$) the solution to the unlabeled multi-robot motion-planning problem in a simple polygon may not always exist.

We devote the remainder of the paper to showing a matching upper bound.
We prove that the unlabeled MRMP problem for unit-disc robots in a \emph{simple} polygon is always solvable for monochromatic separation $\mu = 4$ and bichromatic separation $\beta = 3$, %\karl{should be $\beta=3$?}\danny{indeed $\beta=3$}\irina{no, for the single free space component case we have $\beta=2$ (see Thm.11)}\danny{but this sentence does not seem to confine to a single free space component?!},
assuming that the number of start and target positions match in each free space component.
For the case of a single free space component, we show an even stronger result that the problem is always solvable for $\mu = 4$ and $\beta = 0$.

Specifically, in Section~\ref{sec:single-component} we devise an efficient algorithm for MRMP for $\mu= 4$ and $\beta= 2$ in the case of a single free space component, and then extend it to also work for $\mu= 4$ and $\beta= 0$.
In Section~\ref{sec:multiple-components} we extend the algorithm to the case of a free space with multiple components and $\mu= 4$ and $\beta= 3$.
%Here $\beta = 3$ is necessary to limit the interaction between free space components.
Our algorithm runs in $O(n \log n + mn + m^2)$ time, where $n$ is the size of the polygon, and $m$ is the number of robots.

Our results improve upon the results by Adler et al.~\cite{adler2015efficient}, who describe an algorithm with the same running time that always solves the problem assuming separation of $\mu = \beta = 4$. 
Similarly to their approach, we restrict the robots to move one at a time on a \emph{motion graph} that has the start and target positions as vertices.
Separation of $\mu = \beta = 4$ ensures that the connectivity of the motion graph never changes.
However, in our case, the lower bichromatic separation results in a dynamic motion graph: existence of some edges may depend on whether specific nodes are occupied by the robots.
Furthermore, the lower bichromatic separation in the case of multiple free space components leads to more intricate dependencies between the components.
Nonetheless, we show that there is always an order in which we can process the components, and devise a schedule for the robots to reach their targets.

% ================================================
% =========== Definitions and notation ===========
% ================================================
\section{Definitions and notation}
\label{sec:definitions}
We consider the problem of $m$ indistinguishable unit-disc robots moving in a simple polygonal \emph{workspace} $\W \subset \R^2$ with $n$ edges. 
The \emph{obstacle space} $\mathcal{O}$ is %defined as 
the complement of the workspace, that is, $\mathcal{O} = \R^2 \setminus \W$. 
We refer to points $x \in
\W$ as \emph{positions}, 
%\todo{Mark: above we defined configuration as referring to the whole system, not individual robots. Kevin: I changed the occurrence, where we talk about the system to "configuration of the system". I think, then it is fine to refer to the configuration of a robot as configuration}
and we say that a robot is at position $x$ when its center is positioned at point $x \in \W$. 

For given $x \in \R^2$ and $r \in \R_+$, we define $\D_r(x)$ to be the open disc of radius $r$ centered at $x$. 
% Free space and configurations
The unit-disc robots are defined to be open sets.
Thus, a robot collides with the obstacle space $\mathcal{O}$ if and only if its center is at a distance strictly less than~$1$ from $\mathcal{O}$.
We can now define the \emph{free space} $\F$ to be all positions where a unit-disc robot does not collide with obstacle space, or, more formally, $\F = \{x \in \R^2 \mid \D_1(x) \cap \mathcal{O} = \emptyset \}$.
The free space is therefore a closed set. 
We refer to the connected components of $\F$ as \emph{free space components}. 

As the robots are defined to be open sets, two robots collide if the distance between their positions is strictly less than~$2$.
In other words, if a robot occupies a position $x$ then no other robot can be at a position $y \in \D_2(x)$; we call  $\D_2(x)$ the \emph{aura} of the robot at position $x$.
%Intuitively, the aura of a robot at position $x$ contains all positions $y$ such that a robot centered at $y$ collides with the robot centered at $x$.
In our figures the auras are indicated by dashed circles (see Figure~\ref{fig:definitions_example}).
\begin{figure}[t]
    \centering
    \includegraphics{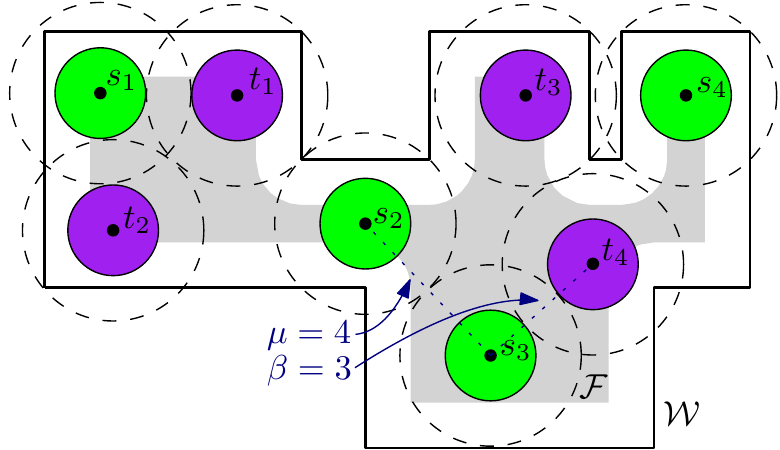}
    \caption{Basic definitions. 
    %using the example from Figure~\ref{fig:robots}.
    The workspace $\W$ is the rectilinear polygon, the free space $\F$ is the inner gray area. The aura of a start or target position is shown as a dashed circle of radius two (for unit-disc robots). The monochromatic separation $\mu = 4$, the bichromatic separation $\beta = 3$.}
    \label{fig:definitions_example}
\end{figure}
%

%We also require that robots do not collide with each other.
\subparagraph{Unlabeled multi-robot motion-planning problem.}
Given a set~$S$ of $m$ start positions and a set~$T$ of $m$ target positions, where $S, T \subset \F$, the goal is to plan a \emph{collision-free} motion for $m$ robots from $S$ to $T$, such that by the end of the motion every target position in $T$ is occupied by some robot.
%We assume that robots placed in all starting positions or in all target positions do not collide, that is, $\D_1(s) \cap \D_1(s') = \emptyset$ for different start positions $s$ and $s'\in S$, and likewise $\D_1(t) \cap \D_1(t') = \emptyset$ for different target positions $t$ and $t'\in T$.
Since the robots are indistinguishable (i.e., unlabeled), it does not matter which robot ends up at which target position.
More formally, we wish to find continuous paths $\pi_i \colon [0, 1] \rightarrow \F$, for $1 \leq i \leq m$, such that $\pi_i(0) = s_i$ and $\{\pi_i(1)\mid 1\leq i\leq m\} = T$.
Furthermore, we require that, at any moment in time $\tau\in[0, 1]$, for all robots $i$, no other robot $j$ is in the aura of robot $i$, $\pi_j(\tau)\not\in\D_2(\pi_i(\tau))$.
In our figures we indicate start positions by green unit discs centered at points in $S$, and target positions by purple unit discs centered at points in $T$.

%Further, let $\partial X$ denote the boundary of some set $X \subset \R^2$.

\medskip
% Subsets and weights
For a subset $Q \subset \F$ of the free space, we use %$S(Q) = \{x \in S \mid x \in Q\}$
$S(Q) = S \cap Q$ to denote the set of start positions that reside in~$Q$, and similarly %$T(Q) = \{x \in T \mid x \in Q\}$
$T(Q) = T \cap Q$ to denote the set of target positions in~$Q$.
%\todo{Mark: Why $s(Q)$ and not $S(Q)$? Same for $t$. And I find "surplus" a nicer term than "weight' In this context. A problem that I would have with surplus is that this quantity is often negative, and we aim at it to be zero. "Surplus" makes it then more difficult to talk about it, e.g., negative surplus.}
We define the \emph{\weight{}} $q(Q)$ as the difference between the number of start and target positions in $Q$, $q(Q) = |S(Q)| - |T(Q)|$.
For each free space component $F_{i}$, we require that $q(F_{i}) = 0$, i.e., there needs to be an equal number of start and target positions.

\medskip
Finally, we state below a few useful properties proven by Adler et al.~\cite{adler2015efficient}.
\begin{lemma}[\cite{adler2015efficient}]\label{lem:adler1}
Each component $F_i$ of the free space is simply connected.
\end{lemma}
\begin{lemma}[\cite{adler2015efficient}]\label{lem:adler2}
For any $x\in\F$, let $F_i$ be the connected component of the free space containing $x$. Then the set $\D_2(x)\cap F_i$ is connected.
\end{lemma}

%% Separation
%Finally, we  distinguish between two types of separability bounds: \emph{monochromatic}, namely between two start positions or between two target positions, denoted by~$\mu$, and \emph{bichromatic}, namely between a start position and a target position, denoted by~$\beta$.
%The two types of separability constraints are illustrated in Figure~\ref{fig:definitions_example}.

%\begin{figure}
%    \centering
%    \includegraphics{images/examples/separation_example.pdf}
%    \caption{A visualization of the two types of separation, namely %monochromatic separation denoted by $\mu$ and bichromatic separation denoted %by $\beta$.}
%    \label{fig:separation_example}
%\end{figure}

% ================================================
% =========== Tighter separation bounds ==========
% ================================================

\section{Tighter separation bounds}
\label{sec:separation-bounds}

In this section we explore the separation between the start and target positions that is necessary for the problem to always have a solution.
%As discussed above, we distinguish between monochromatic ($\mu$) and bichromatic ($\beta$) separation.
%The minimum separation is a constraint we impose on the problem in order to make it easier to solve, since the base problem is likely\todo{HF: what is `likely' meaning?} PSPACE-hard.
We show that, without a certain amount of monochromatic separation ($\mu$) and bichromatic separation ($\beta$), there are instances of the problem that cannot be solved, thus certain separation is necessary for the problem to always have a solution.
%
%the following is redundant and was just said in the contribution section
%Adler et al.~\cite{adler2015efficient} describe an algorithm that can always find a solution to the unlabeled motion planning problem for unit-disc robots in a simple polygonal workspace, assuming a separation of four and that the number of start and target positions is equal in every connected component of the free space.
%They do not make a distinction between monochromatic and bichromatic separation, so their result applies for $\mu = \beta = 4$.
%In their paper, they refer to the positions as \emph{well-separated} to indicate this separation of four. (We will not be using this term in this paper.) 
We first prove that a separation of $\mu = 4$ is necessary. This bound is tight and it improves a previous lower bound of $\mu = 4\sqrt{2} - 2$ ($\approx 3.646$)~\cite{adler2015efficient}.
We then show that $\beta = 3$ is also necessary.

%Importantly, the minimum separation of four is not shown to be tight, since only an example with separation of $\mu < 4\sqrt{2} - 2$ ($\approx 3.646$) is %given for which a solution does not exist. 
%In this section, we tighten the separation bounds, as well as make a distinction between monochromatic and bichromatic separation.
%For both $\mu$ and $\beta$, a lower bound will be given for which a solution does not exists.
% REDUNDANT
%Assuming that the number of start positions in each connected component of the free space is equal to the number of target positions in that component, we aim to find stricter separation bounds that will still guarantee the existence of a solution. 

\begin{figure}[t]
\centering
    \centering
    \includegraphics{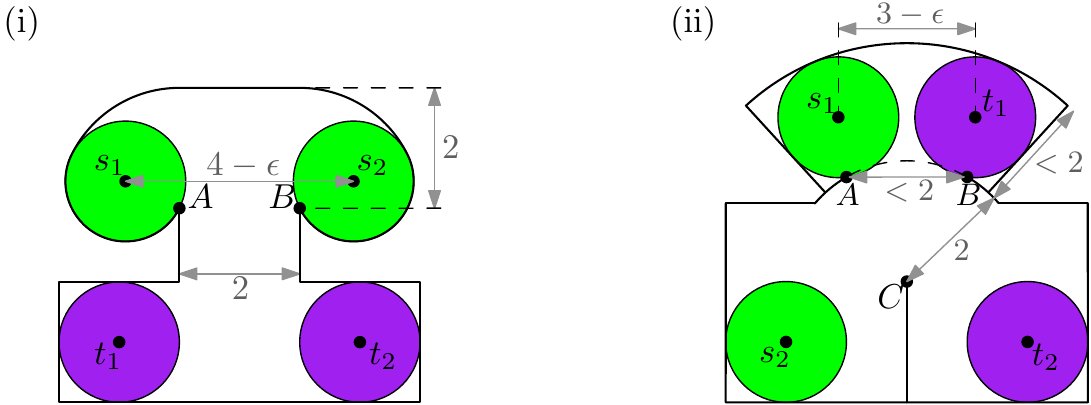}
    \caption{(i) An instance for $\mu = 4 - \epsilon$ with one free space component. The robots are blocking each other from entering the corridor.
	(ii) An instance for $\beta < 3- \epsilon$. The distance $|AB|<2$ is too small for a robot to pass through, thus there are two free space components. The robot in the top component is blocking the one in the bottom component.}
    \label{fig:proof_mu4_and_b3}
\end{figure}

%\subparagraph{Monochromatic separation} $\mu$ refers to the minimum distance between any pair of start positions or between any pair of target positions.
%By the non-collision constraint, $\mu$ should always be at least two for the problem to be valid, otherwise the robots collide will collide at the initial state or at the goal state.
%We present a tight lower bound for $\mu$. 
\begin{restatable}{lemma}{monochromaticLowerBound}
%\begin{lemma}
\label{lem:m4}
For $\mu < 4$ a solution does not always exist, even if the free space consists of a single connected component containing two start and two target positions.
%\end{lemma}
\end{restatable}
\begin{proof}
See Figure~\ref{fig:proof_mu4_and_b3}~(i) for an instance where a solution does not exist when $\mu = 4 - \epsilon$ for some arbitrarily small $\epsilon > 0$. The example uses circular arcs as boundaries, but these can be approximated~\cite{dudley-approx} to obtain a slightly larger simple polygon with positions that are at most $\epsilon/4$ away from a position in the original region.

In the example, two robots $r_1, r_2$ at start positions $s_1, s_2$ need to move through a narrow corridor of width 2 to reach the target positions $t_1, t_2$.
The separation between $s_1$ and $s_2$ is $4 - \epsilon$.
Let points $A, B$ be the endpoints of the corridor closest to $s_1$ and $s_2$, which in the example lie on the boundary of the robots at $s_1$ and $s_2$ respectively.
Clearly, both robots cannot move into the corridor simultaneously, therefore assume, without loss of generality, that $r_1$ moves across the line segment $\overline{AB}$ first.
Thus, for such a solution, robot $r_1$ will need to rotate around point $A$ and then move down the corridor. 

We observe that points $A$ and $B$, the end points of the corridor, must be below the line segment $\overline{s_1s_2}$, given that the corridor has width 2 and the separation between $s_1$ and $s_2$ is less than 4.
Notice that by the triangle inequality we must have that the distance between $A$ and $s_2$ is less than 3.
This means that the aura of $r_2$ at $s_2$ intersects the area swept by $r_1$ as it moves around $A$ and into the corridor.
Furthermore, there is no point in the free space where $r_2$ can move to give space to $r_1$, since any point obstructs the rotation of $r_1$ around $A$.
Therefore, no solution exists for this instance. 
\end{proof}

Thus, for a solution to always exist, a monochromatic separation of $\mu = 4$ is necessary.
Since the problem for $\mu = \beta = 4$ always has a solution, the monochromatic separation is tight.
Hence, we aim to reduce the bichromatic separation $\beta$.
%
%\subparagraph{Bichromatic separation} $\beta$ refers to the minimum distance between any pair of a start and a target position, as defined in Section~\ref{sec:introduction}.
%
\begin{restatable}{lemma}{bichromaticLowerBound}
%\begin{lemma}
\label{lem:b3}
For $\beta < 3$ a solution does not always exist, even if there are only two free space components, each containing one start and one target position.
%\end{lemma}
\end{restatable}
\begin{proof}
See Figure~\ref{fig:proof_mu4_and_b3}~(ii) for an instance where a solution does not exist when $\beta = 3 - \epsilon$ for some arbitrarily small $\epsilon > 0$. As above, we can approximate the circular arcs by polygonal chains.

In the example, there are two connected components of the free space, both containing a start and target position ($s_1, t_1$ and $s_2, t_2$ respectively)
since points $A$ and $B$ are at distance less than $2$.
Let $|AB|=2 - \delta$ for some $\delta > 0$.
Here, we define $\delta$ such that $\delta < 2\epsilon/3$.
In this example, point $A$ lies on the boundary of $\D_1(s_1)$ and points $B$ on the boundary of $\D_1(t_1)$.
From these facts, it follows that $A$ lies to the left of the line segment $\overline{s_1 C}$ and $B$ lies to the right of the line segment $\overline{t_1 C}$.
By the triangle inequality we know that the distance from $s_1$ to $C$ must be less than 3, similarly for the distance from $t_1$ to $C$.

The key characteristic is that no matter where the robot in the top component is, it will block the movement from start to target of the robot in the bottom component.
Since the top arc of the workspace is a semi-circle with center at $C$ and radius less than 4, there is no point in the top component of the free space which does not block the movement in the bottom.
Thus, the robot at $s_2$ can never reach $t_2$, which means no solution exists for this example.
\end{proof}

The lower bound construction for $\beta < 3$ has two free space components with one robot in each.
A robot in the top free space component is blocking the motion of a robot in the bottom component, no matter which position it is in.
Thus, the lower bound is not applicable if the free space has only one component. Indeed, as we  show in the next section, in this case no bichromatic separation is necessary.

% ================================================
% ============ One free space component ==========
% ================================================

\section{A single free space component}
\label{sec:single-component}

In this section we consider the multi-robot motion-planning problem for the case where the free space consists of a single component $F$. 
Initially, for simplicity, we assume $\mu = 4$ and $\beta = 2$.
That is, no start/target position can be inside the aura of another start/target position.
We later modify the algorithm to handle the case with no bichromatic separation.

The algorithm by Adler et al.~\cite{adler2015efficient} uses the separation assumption $\mu = \beta = 4$, and cannot be applied if $\beta < 4$.
Their algorithm greedily moves the robots to the target positions, and may not always be able to find a solution in our case.
Indeed, a pair of a start and a target positions whose auras intersect can possibly block the path for robots who need to go through the intersection of these auras (see Figure~\ref{fig:blocker_and_remote_components}~(i)).
Therefore, in our algorithm we need to handle such blocking positions.

\begin{figure}[b]
\centering
    \centering
    \includegraphics{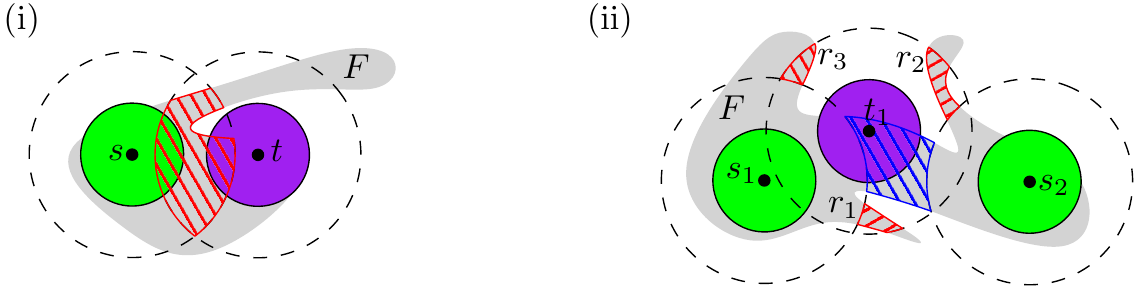}
    \caption{(i) When $\beta < 4$, a robot cannot cross the intersection of the auras of $s$ and $t$ (in red) if either $s$ or $t$ is occupied. (ii) $\Dminus_2(t_1)$ consists of multiple connected components (remote in red and non-remote in blue). %Remote components are shown in red. 
    Remote components $r_1$ and $r_2$ are blocking areas associated with blocker $t_1$.}
    \label{fig:blocker_and_remote_components}
\end{figure}

\subsection{Preliminaries}
\label{sec:preliminaries}
%%Kevin: since we only have 1 component, we don't need S_i and T_i
%%Let $S_i = S(F_i)$ and $T_i = T(F_i)$ be the start and target positions in $F_i$, respectively.
%Let $\Fminus = F \setminus \bigcup_{s \in S} \D_2(s)$ be the portion of the free space $F$ which does not intersect the aura of any start position.
%The subset $\Fminus$ might consist of multiple connected components, since the aura of a start position may intersect the boundary $\partial F$ multiple times (refer to Figure~\ref{fig:blocker_and_remote_components}\irina{new fig?}).
%%\rev{The point you make here is sufficiently easy to understand; however, I found the wording with two occurrences of “multiple connected components” with slightly different meanings a bit confusing.}
%%The boundary $\partial \Fminus$ consists of segments of two types: the \emph{free boundary}, which is shared with the free space boundary, and the \emph{aura boundary}, which is shared with the boundary of the aura of a start position. \todo{Bahar: I have remove this because it is not used and also similar thing is defined in the production of blockable edge}
%Recall that the free space is a closed set and the aura is an open set, thus $\Fminus$ is closed.
\subparagraph{Remote components.}
Let $A(S)=\bigcup_{s \in S} \D_2(s)$ be the union of all auras of the start positions $S$. 
For a target position $t \in T$, let $\Dminus_2(t) = (\D_2(t) \cap F) \setminus A(S)$ be the portion of $F$ within the aura of $t$ minus the auras of the start positions in $S$ (see Figure~\ref{fig:blocker_and_remote_components}(ii)).
Note that even though, by Lemma~\ref{lem:adler2}, $\D_2(x) \cap F$ is always connected for any $x\in F$, the region $\Dminus_2(t)$ may consist of multiple connected components (split by the auras of start positions).
One of these components contains $t$ (shown in blue in the figure).
A component of $\Dminus_2(t)$ that does not contain $t$ is called a \emph{remote component} of $t$ (shown in red in the figure).
Let $R$ be the set of remote components for all target positions in $T$.

%A key characteristic of a remote component of $t$ is that a path between a remote component and $t$ that stays within $\D_2(t)$ has to pass through the aura of at least one start position. This can cause issue when constructing the motion graph, more details follow.

\subparagraph{Blockers and blocking areas.}
Consider the example in Figure~\ref{fig:blocker_and_remote_components}\,(ii).
If $t_1$ is occupied, its remote components $r_1$, $r_2$, and $r_3$ cannot be crossed by a moving robot.
Crossing the remote component $r_3$ can be avoided by moving along its boundary.
However remote components $r_1$ and $r_2$ pose a problem, as they cut the free space, and thus crossing them cannot be avoided.
We call such remote components \emph{blocking areas}.

For a target position $t\in T$, a \emph{blocking area} is a remote component of $t$ that partitions $F$ into multiple components.
If $t$ is associated with at least one blocking area, we refer to $t$ as a \emph{blocker}.
A blocker might have multiple associated blocking areas (as in Figure~\ref{fig:blocker_and_remote_components}\,(ii)).
Let $B \subseteq R$ be the set of blocking areas for all target positions in $T$.

For a blocking area $b \in B$ associated with position $t$,
let the \emph{blocking path} be any path $\pi \subset \D_2(t)$
%\rev{I believe that $\pi \subset D^-_2(t)$ should actually be $\pi \subset D_2(t)$.} 
connecting $b$ to $t$.
By Lemma~\ref{lem:adler2}, $\pi$ exists, and by definition of the blocking area, $\pi$ crosses the aura of at least one start position. %; otherwise, the blocking area would be connected to the component of $\Dminus_2(t)$ which contains $t$.
We further show in the following lemma that this path does not intersect any other blocking area.
\begin{restatable}{lemma}{DisjointBlockingPath}
\label{lem:disjoint-blocking-paths}
For a blocking area $b_x \in B$ and its associated blocker $x$, there exists some blocking path $\pi$ such that $\pi \subset \D_2(x)$ and $\pi$ does not intersect a blocking area $b_y$ of any other blocker $y$.
\end{restatable}
\begin{proof}
First, we show that there exists some blocking path $\pi \subset \D_2(x)$. This follows from (1)~the assumption that the free space consists of one connected component, and (2)~Lemma~\ref{lem:adler2}.
Thus, there exists a path $\pi$ within $\D_2(x)$ connecting $x$ and $b_x$.

%Since, the blocker position and any associated blocking area are inside $\D_2(x)$ and they are also in the same free space component; therefore, they are in $\D^*(x)$. We conclude that there is a path connecting the two inside $\D^*(x)$.\rev{Check the grammar of this sentence.}

Now, we show that $\pi$ does not intersect any other blocking area $b_y$.
The blocking areas $b_x$ and $b_y$ lie inside $\D_2(x)$ and $\D_2(y)$, respectively.
Since $\mu = 4$, the two auras $\D_2(x)$ and $\D_2(y)$ do not intersect.
Thus, there exists a blocking path $\pi$ from $x$ to $b_x$ that stays within $\D_2(x)$ and therefore cannot cross another blocking area $b_y$.
\end{proof}

\subparagraph{Residual components.}
Let $\overline{F} = F \setminus \bigcup R$ %\yoshio{I changed $F \setminus R$ to $F\setminus \bigcup R$.} 
be the portion of the free space $F$ that does not intersect any remote component in $R$.
By definition, a blocking area partitions $F$ into multiple connected components.
%    \mdb{def was wrt $\partial \Fminus_i$, not $\partial F_i$.}
Since some remote components are blocking areas, the region $\overline{F}$ may consist of multiple connected components.
We refer to the connected components of $\overline{F}$ as \emph{residual components}. Furthermore, let $F^* = \overline{F} \setminus A(S)$ be the portion of the free space $F$ that does not intersect either the aura of a start position or a remote component of a target position.
\begin{restatable}{lemma}{FiComplexity}
\label{lem:free-space-complexity}
Given $m$ starting and target positions in a polygonal workspace of size $n$, the subsets $\overline{F}$ and $F^*$ of the free space, and the remote components $R$, all have complexity $O(m + n)$ %\todo{HF: What is the complexity of a subset or a region?} 
and can be computed in $O((m+n)\log(m+n))$ time.
\end{restatable}
\begin{proof}
%The complement of the workspace polygon can be decomposed into $O(n)$ trapezoids by using a vertical decomposition.
%Let $\Sigma$ be the set of these trapezoids and $O(m)$ unit discs centered at the start positions.
%Note that all elements of $\Sigma$ are pairwise disjoint.
%The union of the elements in $\Sigma$ Minkowski-summed with a unit disc is linear in the number of elements plus the complexity of the elements~\cite{kedem1986union}.
%Therefore, the region $\Fminus=F\setminus A(S)$ has complexity $O(m + n)$ and can be generated in $O((m+n)\log(m+n))$ time.
The proof of the lemma follows from the fact that a Minkowski sum has linear complexity in the number of elements plus the complexity of the elements~\cite{kedem1986union}.
Let set $C$ be the union of the obstacles (complement of the workspace) and the robot discs $\D_1(s)$ for all $s\in S$.
Then, the region $\Fminus=F\setminus A(S)$ has complexity $O(m + n)$ and can be generated in $O((m+n)\log(m+n))$ time.

The set of remote components $R(t)$ of a target position $t$ is a subset of $\D_2(t)\setminus A(S)$ restricted to the free space $F$.
Due to the separation bounds, $\D_2(t)\setminus A(S)$ is of constant complexity, and thus $R(t)$ has complexity $O(n)$.
The set of remote components $R$ thus has complexity $O(m + n)$ and can be generated in $O((m+n)\log(m+n))$ time.

%The free space portion of an aura $\D_2$ consists of sections of the free space $F \setminus \bigcup_{x \in S \cup T}\D_2(x)$.
%These sections have complexity of $O(m + n)$, by a similar argument as for the free space region $\Fminus$. Thus the free space region of an aura has complexity of $O(m + n)$.

%Using the same logic, the remote components $R$ and the free space subsets $\overline{F}$ and $F^*$ all use segments from existing sets with complexity $O(m + n)$, therefore their complexity is also bounded by $O(m+n)$ and can be computed in $O((m+n)\log(m+n))$ time.
Finally, the free space subsets $\overline{F}=F\setminus \bigcup R$ and $F^*=\Fminus\setminus \bigcup R$ %\yoshio{I changed $R$ to $\bigcup R$.} 
all use segments from the existing sets with complexity $O(m + n)$, therefore their complexity is also bounded by $O(m+n)$ and can be computed in $O((m+n)\log(m+n))$ time.
\end{proof}

%\subsection{The residual components graph}
%\label{sec:residual_components_graph}

\begin{figure}[b]
    \centering
    \includegraphics{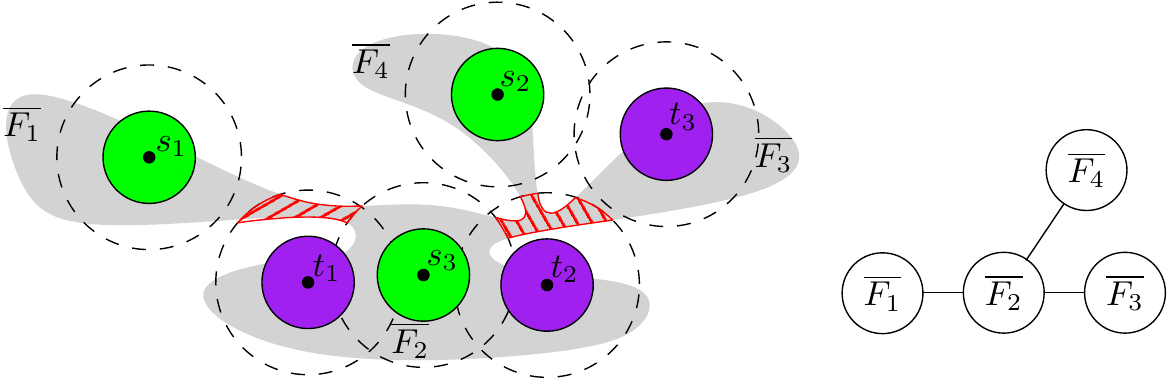}
    \caption{An example with two blockers $t_1$ and $t_2$, and their associated blocking areas shown in red. The corresponding residual components graph $H$ is illustrated on the right side. Note that there is no edge between $\overline{F_{4}}$ and $\overline{F_{3}}$, since the blocker $t_2$ is not located in either of them.}
    \label{fig:example4}
\end{figure}

\subparagraph{Residual components graph.}
We define the \emph{residual components graph} as $H = (V^H, E^H)$ where $V^H$ contains one vertex for each residual component of $\overline{F}$ (see Figure~\ref{fig:example4}).
There is an edge between two vertices $v_1, v_2 \in V^H$ if their respective residual components are separated by a single blocking area $b \in B$ and its associated blocker $t$ resides in the respective residual component of either $v_1$ or $v_2$. 
Although a single blocking area in $B$ can divide $\overline{F}$ into more than two connected components, such a blocking area does not create a cycle in $H$. This is due to the definition of an edge in $H$ which requires the associated blocker to be in one of the two components.
\begin{restatable}{lemma}{AdjacentBlockingAreas}
\label{lem:adjacent_blocking_areas}
Any blocking area $b \in B$ shares a boundary with the residual component containing its associated blocker $t$.
\end{restatable}
\begin{proof}
This follows directly from Lemma~\ref{lem:disjoint-blocking-paths}.
There exists a blocking path $\pi$ which connects $b$ to $t$ and does not cross any other blocking area.
By the definition of a residual component, the blocking area must therefore be adjacent to the component containing $t$.
\end{proof}

\begin{restatable}{lemma}{BlockingGraphConnected}
\label{lem:blocking_graph_tree}
The residual components graph $H$ is a tree.
\end{restatable}
\begin{proof}
First we prove that $H$ is connected. Assume for contradiction that $H$ is not connected.
Then, there must be distinct residual components $x, y \in V^H$ which are not connected by a path in $H$.
Take arbitrary points $p_x \in x$ and $p_y \in y$. 
Since both $p_x$ and $p_y$ lie in $F$ and $F$ is connected, there exists a path $\pi \subset F$ which connects $p_x$ with $p_y$.
Additionally, given the monochromatic separation $\mu = 4$, the blocking areas of distinct blockers do not intersect, and therefore $\pi$ will alternate between a blocking areas and residual components.

Take an arbitrary blocking area $b_t \in B$ that is traversed by $\pi$, which is associated with a blocker target $t$.
Let $v$ and $w$ be the residual components adjacent to $b_t$ that $\pi$ traverses.
We now argue that $v$ and $w$ are connected in $H$.

The blocker $t$ must be in a residual component adjacent to $b_t$ by Lemma~\ref{lem:adjacent_blocking_areas}.
Let $z$ be the residual component containing $t$.
If $z$ is equal to either $v$ or $w$, meaning the blocker $t$ resides in either $v$ or $w$, then by definition there must be an edge between $v$ and $w$ as well, therefore they are connected. 
If $z$ is not $v$ nor $w$, meaning $t$ in a third residual component not equal to $v$ or $w$,
then there must be an edge between $v$ and $z$ and between $z$ and $w$, thus $v$ and $w$ are connected through $z$.
Applying this logic to all blocking areas along $\pi$ between residual components $x$ and $y$, we can conclude that $x$ and $y$ must be connected.

Next, we prove that $H$ is a tree. This follows from how $H$ was constructed:
%$H$ is a tree by construction.
Each blocking area $b\in B$, associated with some blocker $t$, splits $F$ into two or more connected components.
The node of $H$ corresponding to the residual component containing $t$ is connected in $H$ to the other residual components incident to $b$.
Thus, if we consider splitting $F$ and constructing $H$ incrementally, by considering the blocking areas one by one, we insert a new edge for each new node in $H$.
Initially we have one node in $H$ corresponding to the whole $F$.
Thus the number of nodes is one greater than the edges, and $H$ is a tree.
%Assume the graph $H = (V^H,E^H)$ is not a tree.
%Since $H$ is connected, for $H$ not to be a tree it must therefore contain a cycle.
%Thus, there must exist some circular set $v_1, \ldots, v_k$ of distinct vertices in the graph, where $v \in V^H$ and $(v_i, v_{i+1}) \in E$ for $i \in {1, \ldots, k}$ and $k > 2$.
%Given the cycle, there must exist some circular curve $\pi \subset F$ through the residual components $v_1, v_2,\ldots,v_k$, starting and ending in the same position and intersecting at most $k$ blocking areas.
%
%Let $A$ be the area enclosed by $\pi$.
%Given that the workspace $\W$ is simple and $\pi \subset F$, we have that $A \subset F$.
%Given the monochromatic separation $\mu = 4$, the blocking areas $\pi$ intersects are disjoint.
%Since we assumed that $v_1, \ldots, v_k$ correspond to different residual components and $\pi$ should cross each blocking area only once.\rev{Check grammar; this sentence seems incomplete.}
%\mdb{Should the "and" be "we know that"? Or should it be: "Recall that $v_1, \ldots, v_k$ \ldots and that $\pi$ \ldots"?}
%However, in that case $A$ cannot contain more than one residual component, since the disjoint blocking areas can only split $A$ if $\pi$ intersects the blocking area in more than one location.
%We arrive at a contradiction, thus the graph $H$ must be a tree.
\end{proof}

The general idea of our algorithm is to use the residual components graph $H$ to help us split the problem into smaller subproblems.
Using the graph $H$, we will iteratively choose a leaf residual component $\overline{F}_i$ with a non-positive charge (recall that the charge of a component is the number of start positions minus the number of the target positions), and solve the subproblem restricted to that component using its motion graph, which we define shortly.
If afterwards $\overline{F}_i$ will require more robots, they will be moved from a neighboring residual component, ensuring that the blocking area is free for the robots to pass.

\subsection{The motion graph}
\label{sec:motion_graph}
We now introduce the motion graph, which captures \emph{adjacencies} between the start/target positions.
Similarly to~\cite{adler2015efficient}, the underlying idea of our algorithm is to always have the robots positioned on start or target positions and, using the motion graph, to move one robot at a time between these positions until all target positions are occupied. %This restricts robots to $2m$ possible rest positions, which greatly reduces the problem's complexity. %After the motion graph is generated the free space can be ignored. \todo{Irina: removed for clarity}

Recall that for now we assume that the free space consists of one connected component $F$.
For a free space $F$ with start positions $S$ and target positions $T$, 
we define the \emph{motion graph} $G = (V^G, E^G)$, where $V^G = S \cup T$.
The edges $E^G$ in $G$ are of two types: \emph{guaranteed} or \emph{blockable}, which we formally define below.
Guaranteed edges correspond to so called \emph{guaranteed paths}, where a path in the free space $F$ between $u,v\in S\cup T$ is said to be \emph{guaranteed} if it does not intersect the aura of any position other than $u$ and $v$.
%As we will prove in Proposition~\ref{prop:motion-graph}, if $(u,v)\in E^G$ is guaranteed, then there is a guaranteed path connecting $u$ to $v$ in $F$, and thus a robot can always move from $u$ to $v$ (if $v$ is unoccupied).

Unlike guaranteed, blockable edges correspond to paths in $F$ that must cross blocking areas.
%For an edge $(u,v)\in E^G$:
%\begin{enumerate}[$\bullet$]
%\item If $(u,v)$ is a \emph{guaranteed} edge, then there is a path through $F$ between positions $u$ and $v$ that does not intersect the aura of any position other than $u$ and $v$. Hence, a robot can move unobstructed between $u$ and $v$ along that path, as long as the destination is unoccupied.
%
%\item If $(u,v)$ is a \emph{blockable} edge, then any path through $F$ between positions $u$ and $v$ intersects at least one blocking area in $B$. %Additionally, such a path exists.
%
%Note, that existence of a path between $u$ and $v$ does not necessarily imply that there is an edge between $u$ and $v$ in $G$.
%Edges in $E^G$ are of two types. Let $u,v \in V^G$. There is an \emph{guaranteed} edge $(u,v)$ in $E^G$ if there is a path through $F$ between $u$ and $v$  that does not intersect the aura of any positions except $u$ and $v$. Hence, a robot can move unobstructed between $u$ and $v$ along that path, as long as the destination is unoccupied.
%There is a \emph{blockable} edge $(u,v)$ in $E^G$ if there is a path through $F$ between $u$ and $v$ that, in addition to the aura of $u$ and $v$, is allowed to cross the blocking areas in $B$. 
%\mdb{As stated, this seems a superset of the guaranteed edges.}
Our algorithm requires the motion graph $G$ to be connected.
However, as $\beta < 4$, without blockable edges the motion graph may be disconnected (see Figure~\ref{fig:blocking_area_motion_graph}).
Introducing blockable edges ensures that $G$ is connected.
%\end{enumerate}

\subparagraph{Guaranteed edges.}
First, we define the guaranteed edges in $E^G$ and show how to construct corresponding guaranteed paths.
%Let the subset $F^*$ consist of $k$ connected components $\{F^*_{1}, \ldots, F^*_{k}\}$. 
Recall that we define the set $F^*$ to be the free space minus the auras of the start positions and the remote components, $F^* = F \setminus (\bigcup_{s \in S}\D_2(s) \cup \bigcup R)$.  %\yoshio{I changed $R$ to $\bigcup R$.}
Consider a connected component $F^*_{j} \subset F^*$.
Note that $F^*_{j}$ may not be simply-connected, as it may contain holes due to subtracted auras of start positions. 
Abusing the notation, by $\partial F^*_{j}$ we refer to the outer boundary of $F^*_{j}$.
For $\partial F^*_{j}$, we create an ordered circular list $\Lambda_j$ of points along $\partial F^*_{j}$ as follows (see Figure~\ref{fig:motion_graph_example}).

\begin{figure}[t]
    \centering
    \includegraphics{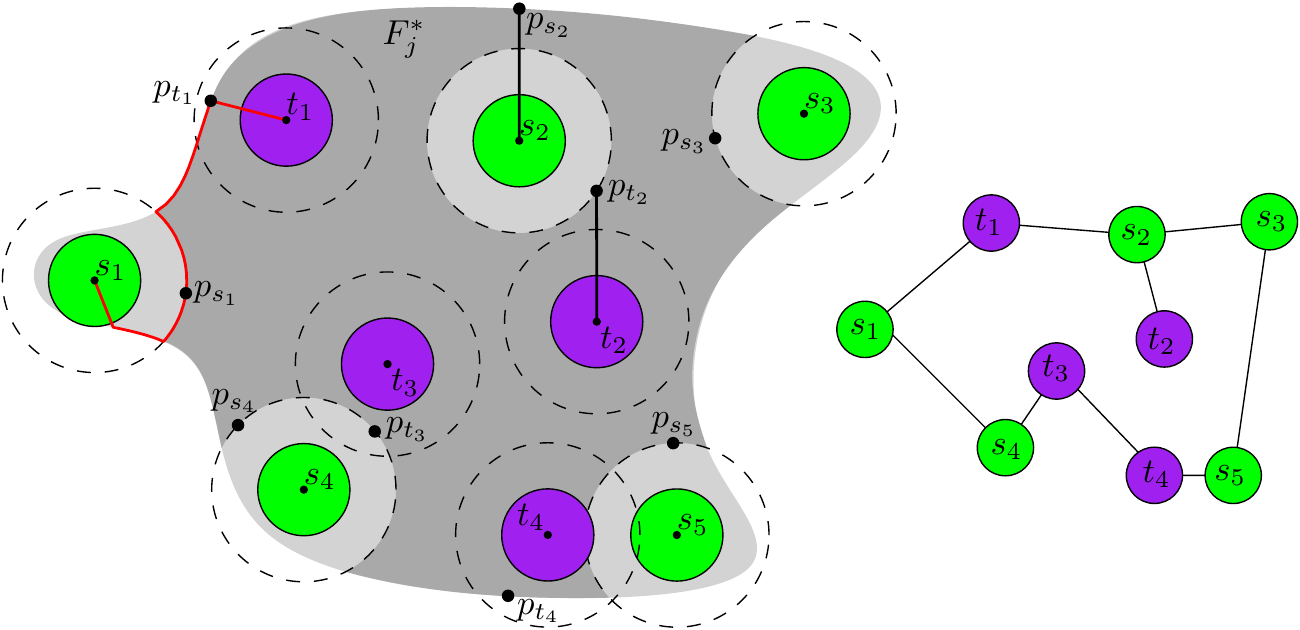}
    \caption{An illustration of generating the guaranteed edges in a single component $F^*_j$. Component $F^*_j$ is shown in dark grey;
    $\Lambda_j$ is $ \langle p_{s_1}, p_{t_1}, p_{s_2}, p_{s_3}, p_{s_5}, p_{t_4}, p_{t_3}, p_{s_4} \rangle$. 
    A path between a pair of adjacent positions is shown in red. The motion graph is shown on the right.}
    \label{fig:motion_graph_example}
\end{figure}

\begin{enumerate}[(i)]
\item For each target position $t \in T\cap F^*_{j}$ whose aura intersects $\partial F^*_{j}$, we pick a set of representative points $P_t$ such that $P_t$ contains one point on each connected component of $\partial F^*_{j} \cap \D_2(t)$.
The points $P_t$ are stored in $\Lambda_j$ based on their ordering along $\partial F^*_{j}$.

%Since $F^*$ does not contain any remote components, there is a path from the target position $t$ to each representative point $p \in P_t$ which stays within $\D_2(t)$ and does not intersect the aura of another start or target position.

\item For each position $x$ which is (1)~either a target position in $F^*_j$ whose aura does not intersect $\partial F^*_j$, or (2)~a start position corresponding to a hole in $F^*_j$, we shoot a ray vertically upwards until it hits either $\partial F^*_{j}$ or the aura of another position $y$.
In the former case the first intersection point $p_x$ is added to $\Lambda_j$ as a representative point of $x$. In the latter case a guaranteed edge is added to $E^G$ connecting $x$ and $y$.

\item Now, consider a start position $s$ whose aura shares a boundary with $\partial F^*_{j}$.
Note that $\partial F^*_{j} \cap \D_2(s)$ is connected.
If we can pick a representative point $p_s$ on $\partial F^*_{j} \cap \D_2(s)$ such that there exists an unobstructed path in $F$ connecting $s$ to $p_s$, then we insert $p_s$ to $\Lambda_j$ based on its ordering along $\partial F^*_{j}$.
Otherwise, if for every choice of $p_s\in \partial F^*_{j} \cap \D_2(s)$ any path connecting $s$ to $p_s$ crosses an aura of some target position $t$, then we add a guaranteed edge to $E^G$ connecting $s$ and $t$ (for every such target position $t$).
Observe, that by the definition of remote components, if a path connecting $s$ to $p_s$ crosses the aura of $t$, it must cross it through the non-remote component of $t$.
Thus, there must exist a guaranteed path connecting $s$ and $t$.
%\irina{I have changed this case. Please compare to the previous version (commented out) and check.}
%\item For each start \configuration{} $s$ whose aura shares a boundary with $\partial(F^*_j)$, we pick a representative point $p_s$ on $\partial (F^*_j) \cap \D_2(s)$.
% If there is an unobstructed path from $s$ to $p_s$, we add $p_s$ to $\Lambda_j$ based on its ordering along $\partial(F^*_j)$.
% Otherwise, if the path crosses the aura of a target $t$ through a non-remote component, an edge is added to the motion graph between $s$ and $t$.
% The remaining case, when any path from $s$ to $p_s$ crosses a remote component will be handled when we add the blockable edges.
\end{enumerate}
Now that $\Lambda_j$ is generated, we add a guaranteed edge to the motion graph between any two nodes in $G$ whose representative points are consecutive in $\Lambda_j$.
%    \mdb{Because of (ii), we should perhaps make it more explicit what the circular sequence is. Essentially,
%    the rays from the floating starting configurations "cut open" $F^*_j$ and turn it into a topological disc. Irina: rays do not cut F*, they only help selecting a point directly above to be added to the list. I modified the figure to avoid giving the impression that we are changing the boundary}
If multiple edges between two vertices and self-loops are generated, we remove them in a post-processing step.
We repeat this procedure for every connected component $F^*_j\subset F^*$.

%------------------
\begin{figure}[t]
    \centering
    \includegraphics{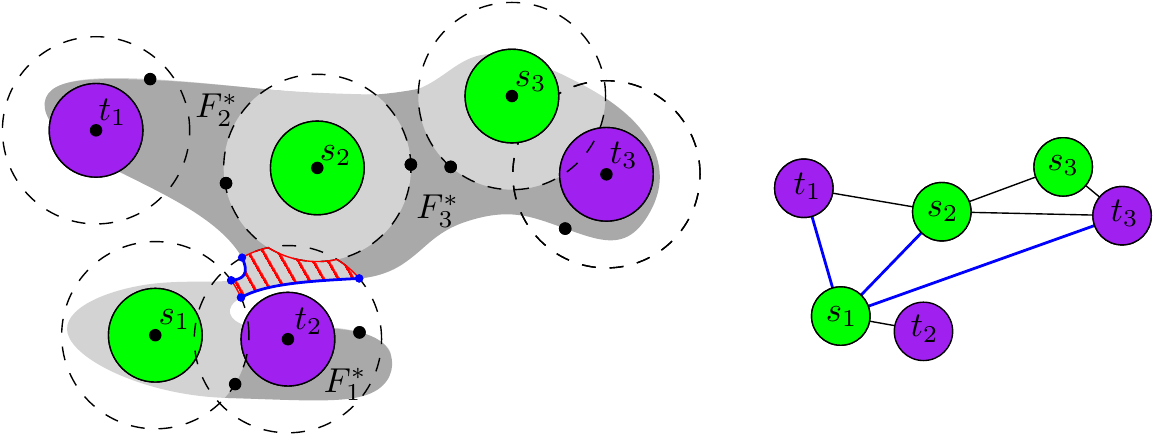}
    \caption{An illustration of a blocking area (in red) with its free boundary (in blue). On the right, the motion graph is shown, with the guaranteed edges (in black) and blockable edges (in blue).}
    \label{fig:blocking_area_motion_graph}
\end{figure}

\subparagraph{Blockable edges.}
For any blocking area $b \in B$ and its associated blocker $t$, each section of $\partial b$ is either shared with (1) the boundary of the aura of $t$, (2) with the boundary of the aura of some start position in $S$, or (3) with $\partial F$.
See Figure~\ref{fig:blocking_area_motion_graph} for an illustration.
We call a section of $\partial b$ which is shared with $\partial F$ a \emph{free boundary} segment of $b$.
For any free boundary segment of $b$ with endpoints $x$ and $y$, we assign a set of \emph{incident} positions in $S\cup T$ to $x$ and to $y$ (see below for details).
We then add a blockable edge to the motion graph between every pair of incident positions of $x$ and $y$ respectively.
Consider an endpoint $x$ of a free boundary segment of $b$.
The set of incident positions of $x$ is defined as follows.
\begin{enumerate}[(i)]
\item If $x$ is also an endpoint of a section of $\partial b$ that is shared with $\partial D_2(s)$ for $s \in S$, then $s$ is the only incident position for $x$.
\item If $x$ is also an endpoint of a section of $\partial b$ that is shared with $\partial D_2(t)$, then $x$ lies on the boundary of a component of $F^*$.
Let that component be $F^*_j$. 
Based on the position of $x$ on $\partial F^*_{j}$, using $\Lambda_j$, we find the predecessor and the successor points of $x$ in $\Lambda_j$.
By construction of $\Lambda_j$, these points are representative of some positions in $S \cup T$.
We select those positions as the incident positions for $x$.
The special case that $\Lambda_j$ is empty, is handled separately and is explained next.
\end{enumerate}
%We repeat this procedure for every free boundary segment of each blocking area in $B$.
%If it results in multiple edges between a pair of nodes in $G$, we remove the duplicates.
%
For the special case when $\Lambda_j$ of $F^*_j$ is empty, if $b$ is the only blocking area incident to $\partial F^*_{j}$, then $F^*_j$ does not contain any position in $S \cup T$, and can be ignored.
Otherwise, if $F^*_j$ is adjacent to another blocking area $b'$ (of some blocker $t'$), then from $x$ we follow $\partial F^*_{j}$ until we reach $\partial b'$ at $x'$, which must be the endpoint of a free boundary segment of $b'$ (see Figure~\ref{fig:blocking_area_two}).
Let $y'$ be the other endpoint of that free boundary segment. 
%By definition $z$ is one of the two endpoints $x$ or $y$ of a free boundary segment of $b$. 
%Without loss of generality, assume that $z = x$.
We now select the incident positions of $y'$ as the incident positions of $x$, i.e., we add a blockable edge between the incident positions of $y$ and those of $y'$.
This results in blockable edges associated with two blocking areas $b$ and $b'$.

\begin{figure}[h]
    \centering
    \includegraphics{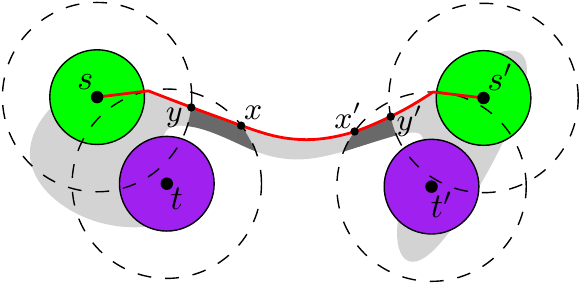}
    \caption{The special case when a blockable edge is added across two blocking areas.}
    \label{fig:blocking_area_two}
\end{figure}

\subparagraph{Translating motion graph edges to free space paths.}
Consider a component $F^*_j \subset F^*$, and let $\Lambda_j$ be the circular list of representative positions constructed for $F^*_j$.
Let $u$ and $v$ be two positions whose representative points are adjacent in $\Lambda_j$. By definition, $(u,v)\in E^G$ is a guaranteed edge, and we claim that there exists a guaranteed path between $u$ and $v$ in $F$.
We construct such a path $\pi_{uv}$ in the following way.
\begin{enumerate}[(i)]
\item First, $\pi_{uv}$ connects $u$ to its representative point $p_u$ by either following an unobstructed path from $u$ to $p_u$ within $u$'s aura, or by following the vertical ray used to generate $p_u$ outside of $u$'s aura.
\item Next, $\pi_{uv}$ connects $p_u$ to the representative point $p_v$ of $v$ by following  $\partial F^*_{j}$.
\item Finally, $\pi_{uv}$ connects $p_v$ to $v$ similarly to (i).
\end{enumerate}

Now consider the case when a guaranteed edge $(u,v)$ is constructed without adding the representative points to $\Lambda_j$.
If $(u,v)$ is constructed according to the case (ii) of the definition of the guaranteed edges, and without loss of generality the vertical ray emanates from $u$, then the guaranteed path $\pi_{uv}$ consists of the vertical segment $u p_u$ and an unobstructed path connecting the representative point $p_u$ to the node $v$ within the aura of $v$.
If $(u,v)$ is constructed in case (iii), and without loss of generality $u\in S$ and $v\in T$, then the guaranteed path $\pi_{uv}$ consists of a path from $u$ until the first intersection with the aura of $v$ and an unobstructed path to $v$ within its aura.

The paths for blockable edges are constructed in the following way.
Each free boundary segment $(x,y)$ of every blocking area $b$ contributes up to four edges between the incident positions of its endpoints.
Consider a blockable edge $(u_x,v_y)$ between an incident position $u_x$ of $x$ and an incident position $v_y$ of $y$.
The corresponding path consists of three parts.
\begin{enumerate}[(i)]
\item From $u_x$ to $x$. This part is generated similarly to the part (i) for guaranteed edges.
\item From $x$ to $y$. This part follows the free boundary of $b$ between $x$ and $y$.
\item From $y$ to $v_y$. This part is generated similarly to the part (iii) for guaranteed edges.
\end{enumerate}

%The following proposition, proved in Appendix~\ref{app:single-component}, 
The following lemmas prove five properties of the motion graph, which will be used to derive the correctness of the algorithm and analyze its complexity.

%\begin{restatable}{proposition}{PropMotionGraph}
%\label{prop:motion-graph}
%The following properties of a motion graph $G$ hold.
%\begin{enumerate}
%    \item There exists a guaranteed path in $F$ for each guaranteed edge in $G$. %\rev{What does it mean for a path to be guaranteed? You have guaranteed edges before, but not guaranteed paths, these are only defined in the appendix in the proof of Proposition 7. This adds extra confusion because path may refer both to a path in the workspace and in the graph.}
%    \item There exists a path in $G$ consisting solely of  guaranteed edges between any two positions inside the same residual component $\overline{F}_j\in\overline{F}$.
%    \item $G$ is connected.
%    \item The number of edges $|E^G|$ in $G$ is bounded by $O(m)$.
%    \item Between any two vertices of $G$, we can find a path in $O(m)$ time, and the corresponding path in the free space has complexity of $O(m + n)$.
%\end{enumerate}
%\end{restatable}
%\begin{proof}
%	The proposition follows from Lemmas~\ref{lem:guaranteed_paths}--\ref{lem:motion_graph_complexity} below.
%\end{proof}

\begin{restatable}{lemma}{GuaranteedPaths}
\label{lem:guaranteed_paths}
For a motion graph $G$, there exists a guaranteed path in $F$ for each guaranteed edge in $G$.
\end{restatable}
%\GuaranteedPaths*
\begin{proof}
	We argue that each portion of a path created in $F^*_{j}$ between two positions $x$ and $y$ is guaranteed, i.e., it does not cross the aura of any position besides $x$ and~$y$. %\rev{Here you define what a guaranteed path is (move that to the main body).}
	
	The motion from a position $x$ to its representative point $p_x$ cannot be blocked by its construction.
	Given bichromatic separation $\beta=2$, the position $x$ itself is not inside the aura of another position.
	If the path from $x$ to $p_x$ crosses some other aura at any point, then the construction will update $p_x$ to the intersection and connect $x$ to this position.
	Thus, the path from a position to its representative point cannot be blocked by another position.
	
	Given the monochromatic separation and the definition of $F^*$, a point on the boundary of $F^*_{j}$ can only ever be in the aura of a single target position.
	Additionally, for each segment of the boundary that intersects the aura of a target we choose a representative point.
	Therefore, the portion of the boundary between representative points that are adjacent on $\Lambda_j$ can only intersect the aura of those two representative points.
	As a result, the motion along the boundary of $F^*_{j}$ between adjacent representative points $p_x$ and $p_y$ cannot intersect the aura of a third position not equal to $x$ or $y$, and thus it is guaranteed.
\end{proof}

\begin{restatable}{lemma}{UnblockedPath}
\label{lem:unblocked_path}
For a motion graph $G$, there exists a path in $G$ consisting solely of  guaranteed edges between any two positions inside the same residual component $\overline{F}_j\in\overline{F}$.
%there exists an unblocked path in the motion graph between two positions that are inside the same residual component.
\end{restatable}
%\UnblockedPath*
\begin{proof}
	Take two positions $x, y \in V^G$ that are inside the same residual component, as defined in Section~\ref{sec:preliminaries}.
	This means that there exists a path $\pi$ through the free space which does not cross any blocking area.
	If $x$ and $y$ reside in the same component of $F^*$, then by Lemma~\ref{lem:guaranteed_paths} there exists a guaranteed path and we are done.
	Otherwise, if $x$ and $y$ are in different components of $F^*$, then $\pi$ must cross auras of some start positions 
	%\yoshio{some start positions $\to$ the auras of some start positions?} 
	in $S$ that split $F^*$ into multiple components.
	Let $s_1, \ldots, s_k$ be the start positions whose auras $\pi$ intersects.
	Then all adjacent positions in the sequence $x, s_1, \ldots, s_k, y$ will share a residual component that they either reside in or have a boundary with.
	By Lemma~\ref{lem:guaranteed_paths}, there must therefore exist a guaranteed path between each adjacent position in this sequence.
	Using those individual paths, the vertices $x$ and $y$ are connected with a path that only uses guaranteed edges.
\end{proof}

\begin{restatable}{lemma}{MotionGraphConnected}
	The motion graph $G$ is connected.
\end{restatable}
%\MotionGraphConnected*
\begin{proof}
	By Lemma~\ref{lem:unblocked_path} we know there exists a path in the motion graph between any positions in the same residual component.
	By Lemma~\ref{lem:blocking_graph_tree}, the residual components graph $H$ is connected.
	The procedure done for each blocking area will ensure that two residual components that are adjacent in $H$ also have edges in the motion graph between two positions in either component.
	Combining these results, the motion graph $G$ must be connected.
\end{proof}

\begin{restatable}{lemma}{MotionGraphEdges}
	\label{lem:motion_graph_edges}
	The number of edges $|E^G|$ in the motion graph $G$ is bounded by $O(m)$.
\end{restatable}
%\MotionGraphEdges*
\begin{proof}
	By construction of $H$, we know that edges are created for every connected component of $F^*$ and then additional blockable edges are added for each blocking area. We argue that the number of both these two types of edges is bounded by $O(m)$.
	
	For each component $F^*_{j} \subset F^*$ the number of edges created is bounded by the target positions that reside in it plus the start positions that share borders of their auras with it.
	A start position can only add one or two edges per component of $F^*$ it borders.
	The edges that a target position contributes is slightly more complicated to analyze, since for a target $t$ we add a representative point for each connected component of $\partial F^*_{j} \cap \D_2(t)$.
	
	However, the number of connected components of $\partial F^*_{j} \cap \D_2(t)$ is constant, since a connected component of $\partial F^*_{j} \cap \D_2(t)$ will have to intersect $\partial \D_2(t)$ in two points.
	Each point $x \in F \cap \partial \D_2(t)$ in the free space requires that $\D_1(x) \cap \mathcal{O} = \emptyset $.
	We can only fit a constant number of unit circles on $\partial \D_2(t)$, therefore there can only be a constant number of segments of $\partial \D_2(t)$ that are in the free space $F$.
	Thus, the number of connected components of $\partial F^*_{j} \cap \D_2(t)$ is constant.
	
	By a similar reason, a start position only borders a constant number of components of $F^*$.
	So the number of edges created for all components of $F^*$ is bounded by $O(m)$.
	
	A blocking area adds at most four edges to $E^G$ per segment of free space boundary.
	By definition, a blocking area will intersect the free space boundary in at least two distinct connected components.
	However, a target can only intersect the free space boundary in a constant number of connected components.
	Therefore, a blocking area has a constant amount of free space boundary segments and all blocking areas in $B$ will add $O(m)$ edges.
\end{proof}

\begin{restatable}{lemma}{MotionGraphPaths}
	\label{lem:motion_graph_paths}
	Between any two vertices of the motion graph $G$, we can find a path in $O(m)$ time, and the corresponding path in the free space has complexity of $O(m + n)$.
\end{restatable}
%\MotionGraphPaths*
\begin{proof}
	From Lemma~\ref{lem:motion_graph_edges}, we know that the number of edges are bounded by $O(m)$.
	Using a simple path-finding algorithm, like a breadth-first search (BFS), a path can be found between two vertices in $O(|V^G| + |E^G|) = O(m)$.
	
	Any path between two vertices in the motion graph will take at most $2m - 1$ %\yoshio{$m-1$ $\to$ $2m-1$ (since we have at most $2m$ vertices)?}
	edges.
	The generated path in the free space, after the translation, will consist of the free space boundary and the portion between a position and a representative point.
	Crucially, a section of the free space boundary will only be taken once in any path, while the path between a position and its representative point is taken either once or twice (e.g. potentially in both directions).
	Therefore, using Lemma~\ref{lem:free-space-complexity} the entire path through the free space will have complexity $O(m + n)$.
\end{proof}

\begin{restatable}{lemma}{MotionGraphComplexity}
	\label{lem:motion_graph_complexity}
	The motion graph $G$ can be created in $O(mn + m^2)$ time.
\end{restatable}
%\MotionGraphComplexity*
\begin{proof}
	As argued in Lemma~\ref{lem:motion_graph_edges}, the number of edges in the motion graph is bounded by $O(m)$.
	Each edge is computed using simple procedures, that are dependent on the components of $F^*$ and the blocking areas $B$.
	By Lemma~\ref{lem:free-space-complexity}, the subset $F^*$ of the free space and the set of blocking areas $B$ both have complexity $O(m + n)$.
	Thus, the entire procedure is bounded by $O(mn + m^2)$.
\end{proof}

\subsection{The algorithm}
\label{sec:divide_conquer}
We are now ready to describe our algorithm.
We use the residual components graph $H$, which is a tree, in order to split the problem into smaller subproblems, and recursively solve them.
Using $H$ we select a particular residual component of the free space, and solve the subproblem restricted to it using the motion graph.
%Proposition~\ref{prop:motion-graph} 
Lemmas~\ref{lem:guaranteed_paths}--\ref{lem:motion_graph_complexity} will help us ensure that such reconfiguration is always possible.
%Then we remove the corresponding vertex and its incident edges from $H$ and recurse on the remaining subtrees.
One key point is to select a vertex of $H$, such that, after solving the subproblem in the corresponding residual component, no robots need to move across the incident blocking areas.
That is, we need to choose the residual components in such an order that we can ignore blockers in the solved residual components.

Recall that a charge $q(Q)$, for some $Q\subseteq F$, is the difference between the numbers of the start positions and the target positions in $Q$.
Initially, if there is an edge $e \in E^H$ such that removing $e$ splits $H$ into two subtrees with zero total charge each, then we remove $e$ from $H$ and recurse on the two subtrees.

%If no such edge exists, then there is a vertex $v\in V^H$ whose incident edges each split $H$ into two subtrees of non-zero \weight{}. Such a vertex requires robots to move across each of its incident edges in order to balance the number of robots and unoccupied targets on the two subtrees connected by each edge.

Let us now assign an orientation to the edges of $H$ in the following way.
For each edge $e=(u,v)$, let $H_u$ and $H_v$ be the two trees of $H\setminus\{e\}$ containing $u$ and $v$ respectively.
We orient $e$ from $u$ to $v$ if $q(H_u)>0>q(H_v)$, and from $v$ to $u$ if $q(H_u)<0<q(H_v)$.
%Let $D = (V^D, E^D)$ be a directed graph such that $V^D = V^H$, and for each undirected edge $e$ between $u$ and $v$ in $E^G$ there is a directed edge in $E^D$: the direction is from $u$ to $v$ if after removing $e$ the subtree containing $u$ has a positive \weight{}, and it is from $v$ to $u$ if it has a negative \weight{}.
%
%If there is an edge that splits the problem evenly, then these are first removed to split $H$ into as many smaller subproblems as possible.
%Otherwise, the directed graph $D = (V^D, E^D)$ is created by directing all edges in $H$ from positive \weight{} subgraphs to negative \weight{}. Here, a sink node $\sigma$ is then selected and solved.
%

\begin{restatable}{lemma}{sink}
\label{lem:sink}
There exists at least one sink vertex in directed graph $H$.
\end{restatable}
\begin{proof}
Since there are no edges that split $H$ into zero-\weight{} subgraphs, every edge in $H$ has a defined direction.
The residual components graph $H$ is a tree, by Lemma~\ref{lem:blocking_graph_tree}, and thus $H$ is a directed acyclic graph (DAG) and must contain at least one sink vertex.
\end{proof}

Using Lemma~\ref{lem:sink}, our algorithm selects a sink node $\sigma$ of $H$.
The respective residual component $\overline{F}_\sigma$ is solved as follows.
First, all robots inside~$\overline{F}_\sigma$ are moved to unoccupied target positions.
Since all incident edges of $\sigma$ in $H$ are directed inwards, each edge requires one or more robot(s) to move into $\overline{F}_\sigma$.
The exact number can be computed from the charges of the subtrees of the adjacent residual components.
We then move the required number of robots from the adjacent residual components to $\overline{F}_\sigma$ over the corresponding blocking areas.
Consider the blocker $t$ associated with a blocking area $b$ incident to $\overline{F}_\sigma$.
If $t$ is occupied before the charge of $\overline{F}_\sigma$ becomes zero, then $t$ has to reside in~$\overline{F}_\sigma$, as the adjacent residual components were not yet processed by the algorithm.
Then we move the robot from $t$ to another unoccupied target in $\overline{F}_\sigma$, and the now unoccupied blocker position $t$ is the last target to become occupied by a robot moving across $b$.

Once all target positions of $\overline{F}_\sigma$ are filled, $\sigma$ and its incident edges are removed from $H$, and we recurse on the remaining subtrees.
Due to the way we select $\sigma$, and the number of robots that are moved into $\overline{F}_\sigma$, each subtree has a total zero-charge. 
See Algorithm~\ref{alg:single_comp_divide} for the pseudocode.

\begin{algorithm}[t]
	\centering
	\begin{algorithmic}
		\Procedure{MotionPlanningDivideConquer}{$\W, S, T$}
		\State Calculate free space $F$
		\State Compute blocking areas and create residual components graph $H = (V^H, E^H)$
		\State Create the motion graph $G = (V^G, E^G)$
		\If{there exists an edge $e \in E^H$ which divides $H$ into zero-\weight{} parts}
		\State Cut $H$ on $e$ and recurse on parts
		\Else
		\State Create a directed graph $D = (V^D, E^D)$ from $H$
		\State Find a sink node $\sigma \in V^D$
		\For{each start $s \in \sigma$}
		\State Find the nearest unoccupied target $t \in \sigma$
		\State Move a (chain of) robot(s) to occupy $t$ and free $s$
		%\todo[inline]{YO: Did we define a pebble move somewhere?}
		\EndFor 
		\For{each edge $e \in E^D$ that points to $\sigma$}
		\For{each robot at start $s$ that needs to move in across $e$}
		\State Find an unoccupied target $t \in \sigma$
		\If{a blocker $t_b \in \sigma$ blocks movement from $s$ to $t$}
		\State Move a (chain of) robot(s) to occupy $t$ and free $t_b$
		\State Move a (chain of) robot(s) to occupy $t_b$ and free $s$
		\Else
		\State Move a (chain of) robot(s) to occupy $t$ and free $s$
		\EndIf
		\EndFor
		\EndFor
		\State Remove $\sigma$ from $H$ and $D$ and recurse on subproblems
		\EndIf
		\EndProcedure
	\end{algorithmic}
	\caption{%Pseudo code for t
		The divide-and-conquer algorithm for a single free space component~$F$.}
	\label{alg:single_comp_divide}
\end{algorithm}

\begin{restatable}{theorem}{DivideConquerCorrectness}
\label{th:divide_conquer_correctness}
\label{lem:divide_complexity}
%For a connected free space component $F \subset \F$ containing an equal number of start and target positions, 
When the free space consists of a single connected component, our algorithm finds a solution to the unlabeled motion planning problem for unit-disc robots in a simple workspace, assuming monochromatic separation $\mu = 4$ and bichromatic separation $\beta = 2$. This takes $O(n \log n + mn + m^2)$ time.
\end{restatable}
\begin{proof}
We argue the correctness using the induction on the number of vertices of $H$.

If $H$ consists of a single residual component $\overline{F_j}$, then the algorithm treats $\overline{F_j}$ vacuously as a sink.
All target positions are then filled greedily by finding the nearest unoccupied target for each start, and by Lemma~\ref{lem:unblocked_path} there is an unblocked path in the motion graph between any two such positions.
Thus, it always finds a valid solution.

Otherwise, assume that $H$ consists of more than one residual component.
We assume that the algorithm is correct for any residual components graph $H'$ with fewer vertices than $H$ (induction hypothesis).
We now show that the algorithm will reduce $H$ into smaller subproblems which can be solved.

If there exists an edge $e \in E_H$ which divides $H$ into zero-\weight{} subtrees the algorithm cuts~$e$ from $H$ and recurses on both subproblems.
Since both subtrees have strictly fewer vertices and an equal number of start and target positions, both subtrees can be solved according to the induction hypothesis.

If no such edge exists, the algorithm finds a ``sink'' residual component $\sigma$ where all adjacent edges in $H$ require robots to move into $\sigma$.
By Lemma~\ref{lem:sink}, there always exists such a sink component.
Since $\sigma$ has only edges pointed inwards, it means that all edges require one or more robots to move into $\sigma$ in order to fill its target positions.
This means the free space associated with $\sigma$ will have at least one more target position compared to start positions for every adjacent edge. 

Before moving robots into $\sigma$, the robots that already reside in $\sigma$ will be moved to target positions.
By Lemma~\ref{lem:unblocked_path}, there exists a path in the motion graph that cannot be blocked and thus this is always possible.
Afterwards, the required number of robots are moved in across every edge adjacent to $\sigma$ in $H_i$.
These paths can cross blocking areas, however, only targets inside $\sigma$ can be occupied.
Thus, if the robots have to cross the blocking area of an occupied blocker $b_t$ inside $\sigma$, the algorithm will first move the robot at $b_t$ to the original target position (which is always possible by Lemma~\ref{lem:unblocked_path}) and then move the robot outside $\sigma$ to $t_b$ instead. 

Once $\sigma$ is completely solved, the algorithm will remove $\sigma$ and all its adjacent edges from $H_i$.
All remaining connected components are strictly smaller and have an equal number of start and target positions, thus we have assumed the algorithm is able to solve them.

For the running time, similarly to before, calculating the free space component $F_i$ and the blocking areas $B_i$ is bounded by $O((m+n)\log(m+n))$ by Lemma~\ref{lem:free-space-complexity}.
The motion graph can then be calculated in $O(mn + m^2)$ by Lemma~\ref{lem:motion_graph_complexity}.
The residual components graph can be easily computed from the calculated components, and can similarly be generated within $O((m+n)\log(m+n))$.

In total, the algorithm will calculate a path for each target position, which is bounded by $O(m)$ by Lemma~\ref{lem:motion_graph_paths}, thus calculating all paths can be done in $O(m^2)$.
A path sometimes requires the additional movement of a robot at a blocker position, but this does not influence the $O(m)$ bound.
The resulting paths for each robot will have complexity $O(m+n)$, therefore the entire motion plan generated will have complexity $O(mn)$.

The number of iterations for the recursive algorithm is bounded by $O(m)$, since it always either occupies a target position or removes an edge from $H_i$.
Since $H_i$ is a tree by Lemma~\ref{lem:blocking_graph_tree}, the number of edges is bounded by $O(m)$.
Detecting an edge in $H_i$ to remove can be done efficiently in $O(m)$ by rooting the graph $H_i$ and for each vertex storing the \weight{} of the subtree rooted at that vertex.
An edge that splits $H_i$ into zero-\weight{} parts will then correspond to an edge where the subtree rooted at the ``child'' vertex has a \weight{} of zero.
Similarly, storing the \weight{} for each subtree of $H_i$ allows us to find the sink vertex $\sigma$ efficiently.
After $\sigma$ is solved and removed, updating the \weight{} can be done in $O(m)$.

Thus, the algorithm can find a solution in $O((m+n)\log(m+n) + mn + m^2)$.
\end{proof}

%\subsection{No bichromatic separation}
%\label{sec:extension}

%\rev{The paragraph following Thm 9 is unclear.}
We now show how to extend our approach to $\beta=0$.
In the above algorithm, we use the bichromatic separation $\beta = 2$ to ensure that at every moment in time any subset of nodes of the motion graph can be occupied by robots.
If $\beta<2$ we can no longer assume that any start position and any target position can be occupied at the same time.
%The bichromatic separation $\beta = 2$ has only been used while creating the motion graph, to show that all target positions are in residual components of the free space after taking the complement of the start positions. 
Nevertheless, even when $\beta=0$, observe that, due to $\mu=4$, for a pair of start and target positions $s_i$ and $t_j$ such that $|s_i t_j|<2$, no other target position $t_k$ can lie in $\D_2(s_i)$, and no other start position $s_\ell$ can lie in $\D_2(t_j)$.
%When $\beta = 0$, since we still assume $\mu = 4$, positions with a distance of less than two to each other only come in pairs of one start and one target position.
Thus, there is a guaranteed path from $s_i$ to $t_j$.
%We can show that there is an unobstructed path between such two positions.
This can be exploited to adjust the motion graph and the algorithms for $\beta = 0$.
Specifically, for each pair of such $s_i$ and $t_j$, we create a single target node in our motion graph, we move the robot from $s_i$ to $t_j$, and adjust our algorithm to work for the case of different number of start and target nodes in the motion graph.

\begin{lemma}
	\label{lem:start-target-pairs}
	For any position $u \in S \cup T$, there is at most one other position $v \in S \cup T$, $v \neq u$, that resides inside $\D_2(u)$.
\end{lemma}

\begin{proof}
	Assume for contradiction that there exists some position $u$ which has more than one start/target position inside its aura.
	Assume without loss of generality that $u$ is a start position.
	Since $\mu = 4$, we know that the positions in $\D_2(u)$ must be target positions. Let $v, w$ be two such target positions.
	Using the triangle inequality, we have that $d(v, w) \leq d(v, u) + d(u, w) < 2 + 2 < 4$, given that the aura is defined as an open set. This contradicts the monochromatic separation $\mu = 4$, since $v$ and $w$ are both targets.
\end{proof}

\begin{lemma}
	\label{lem:path_b0}
	For any two positions $x, y \in S \cup T$ such that $y \in \D_2(x)$ and $x \neq y$, there exists a path $\pi \subset F \cap \D_2(x) \cap \D_2(y)$ connecting $x$ and $y$.
\end{lemma}

\begin{proof}
	The proof is similar to Lemma 2 in Adler et al.~\cite{adler2015efficient}.
	Take $x, y \in S \cup T$ such that $y \in \D_2(x)$ and $x \neq y$.
	By this lemma, there exists a path $\pi' \subset F \cap \D_2(x)$ connecting $x$ and $y$.
	Assume the line segment $\overline{xy}$ does not lie in the free space, since otherwise the path $\pi = \overline{xy}$ stays within $F \cap \D_2(x) \cap \D_2(y)$.
	Take positions $x', y'$ on $\overline{xy}$ such that $x',y' \in F$ and the distance $||x-y||$ is minimized.
	Let $A$ be the area enclosed by $\pi \cup \overline{xy}$ and $A^* = A \setminus F$ be the part of $A$ which lies in the obstacle space.
	We claim that $A^* \subset \D_2(x) \cup \D_2(y)$.

\begin{figure}[t]
	\centering
	\includegraphics{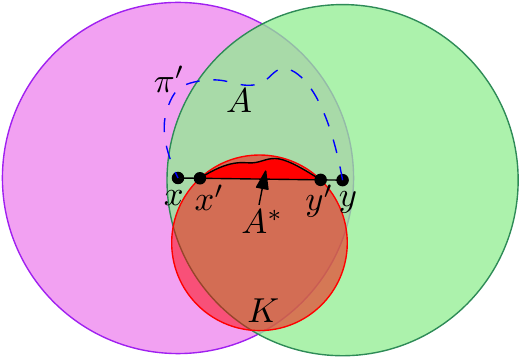}
	\caption{Illustration for Lemma~\ref{lem:path_b0}. The light green and pink discs represent the auras of $y$ and $x$ respectively, while $A^*$ is contained within the red unit-disc $K$.}
	\label{fig:path_b0}
\end{figure}

	Assume without loss of generality that $y$ lies directly to the right of $x$ and that $\pi$ above $\overline{xy}$.
	Let $K$ be the unit circle which intersects both $x$ and $y$. Note that $K$ must lie below $\overline{xy}$.
	The region $A^*$ must then be entirely enclosed within $K \cap A$, which must be within $\D_2(x) \cup \D_2(y)$.
	Thus, a path consisting of $\overline{xx'}$, $\partial A^*$, and $\overline{y'y}$ connects $x$ and $y$ through $F_i$ and stays within $\D_2(x) \cap \D_2(y)$.
	See Figure~\ref{fig:path_b0} for an illustration.
\end{proof}

We are now ready to prove the following theorem.

\begin{restatable}{theorem}{Theorem}
\label{th:no_bichromatic_separation}
%For a single connected free space component $F_i \subset \F$ containing an equal number of start and target positions, 
When the free space $\F$ consists of a single component,
the algorithm finds a solution to the unlabeled motion planning problem for unit-disc robots in a simple workspace, assuming monochromatic separation $\mu = 4$. This takes $O(n \log n + mn + m^2)$ time.
\end{restatable}
\begin{proof}
By Theorem~\ref{th:divide_conquer_correctness}, the algorithm finds a solution assuming bichromatic separation $\beta = 2$.
We will now show how to adjust the algorithm and the construction of the motion graph once this assumption no longer holds.

Let $Q \subseteq S$ be the set of start positions for which there exists a target position in their aura.
By Lemma~\ref{lem:start-target-pairs}, we know that each $s \in Q$ has a unique target position $t \in T$ in its aura, which we will denote with $t_s$.
Let %$T_Q = \bigcup_{s \in Q} t_s$
$T_Q = \{t_s \mid s \in Q\}$ 
%\yoshio{I changed $\bigcup_{s \in Q} t_s$ to $\{t_s \mid s \in Q\}$.}
%\todo{YO: The notation collides with the set of targets in $Q$.} 
be the set of target positions residing in the aura of some start in $Q$.
For the construction of the motion graph we ignore all target positions in $T_Q$.
The start positions in $Q$ will be handled the same as regular start positions.
This will leave all start positions and all targets with bichromatic separation of $\beta = 2$, meaning the motion graph can be generated as before. 

For the algorithm the only adjustment is that the start positions in $Q$ need to be treated as both a start and target.
%% Yoshio: The following part is commented out
%% since we don't talk about the matching algorithm.
%For the matching algorithm (see Section~\ref{sec:matching_algorithm}), this simply means the initial matching $M$ should match each node associated to start position $s \in Q$ to itself.
%Removing cyclical blockings remains unchanged, since the path between a node matched to itself can never be blocked thus also not be part of a cyclical blocking (note that a swap which such a node can still be made).
For the divide-and-conquer algorithm, the start positions only play a role when solving the problem within each residual component.
However, we can again simply treat the start position as a target when solving a component.
The start position can never act as a blocker by definition, and thus no additional problems arise.

After the initial algorithms are finished, the robots will be at target positions $T \setminus T_Q$ and at all start positions in $Q$.
What remains is then to move the robots from the start positions in $Q$ to their matched targets in $T_Q$ such that at the end each target in $T$ is occupied.
By Lemma~\ref{lem:path_b0} for each start position $s \in Q$ there exists a path $\pi$ from $s$ to $t_s$ which stays within $F \cap \D_2(s) \cap \D_2(t_s)$.
Given the monochromatic separation $\mu = 4$, $\pi$ does not cross the aura of another position in $S \cup T$ besides $s$ and $t_s$.
Therefore, we can move the robot at $s$ to $t_s$ across $\pi$ without interference from another position.
Doing this for all starts in $Q$ will result in all target positions in $T$ being occupied.
\end{proof}

% ================================================
% =========== Many free space components =========
% ================================================

\section{Multiple free space components}
\label{sec:multiple-components}
In this section we consider the case where the free space $\F$ consists of multiple connected components.
Since a separation of $\beta = 3$ is necessary to guarantee a solution, we now assume separation bounds of $\mu = 4$ and $\beta = 3$.

Within each free space component we can use the algorithm from Section~\ref{sec:single-component}.
However, paths, that are otherwise valid, may be blocked by a robot from another component.
Consider an example in Figure~\ref{fig:blocking-components}.
Robot at position $s_{2}$ is blocking the movement of a robot from $s_{1}$ to $t_{1}$ in another free space component.
In this example, there is a simple solution: Move the robot away from $s_2$ to $t_2$ in the upper component, before moving the robot from $s_1$ to $t_1$ in the lower component.
In the following we prove that there always exists an order on the free space components such that the motion planning problem can be solved by solving the problem component by component in that order.
\begin{figure}[h]
    \centering
    \includegraphics{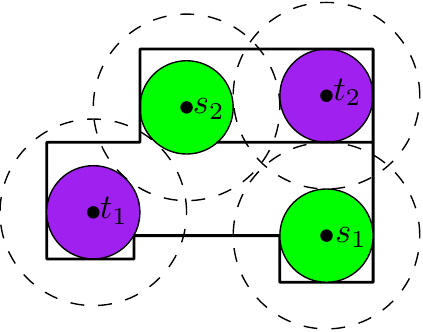}
    \caption{An example of a position ($s_2$) blocking movement ($s_1$ to $t_1$) in another free space component.}
    \label{fig:blocking-components}
\end{figure}

Our approach generalizes a procedure by Adler et al.~\cite{adler2015efficient} for determining an order in which the free space components can be solved. We first describe their procedure for the case $\mu = \beta = 4$.

Let $F, F'$ be two distinct components of $\F$, and let $x \in F$ be such that $\D_2(x) \setminus F' \neq \emptyset$.
We then call $x$ an \emph{interference position} from $F$ to $F'$, and define the \emph{interference set} from $F$ to $F'$ as $\IS(F, F') = \{x \in F : \D_2(x) \cap F' \neq \emptyset\}$.
We also define the \emph{mutual interference set} of $F, F'$ as $\mIS(F, F') = \IS(F, F') \cup \IS(F', F)$.
Intuitively, an element of the interference set from $F$ to $F'$ is a point in $F$ which, when a robot occupies it, could block a path in $F'$, and the interference set is the set of all such points.
The mutual interference set of $F, F'$ is the set of all single-robot positions in either component which might block a valid single-robot path in the other component.

To obtain the ordering of free space components, a directed graph representing the structure of $\F$ is defined, called the directed-interference forest $G = (V, E)$, where the nodes in $V$ correspond to the components $F$.
We add the directed edge $(F, F')$ to $E$ if either there exists a start position $s \in S$ such that $s \in \IS(F, F')$, or there exists a target position $t \in T$ such that $t \in \IS(F', F)$.
Intuitively, a directed edge $(F, F')$ shows that if $F$ is not solved before $F'$, interference will occur. 

In Lemma 3 of the paper by Adler et al.~\cite{adler2015efficient}, it is shown that for any mutual interference set $\mIS(F, F')$ and any two positions $x_1, x_2 \in \mIS(F, F')$ we have $\D_2(x_1) \cap \D_2(x_2) \neq \emptyset$.
Together with the separation constraints $\mu = \beta = 4$ that is assumed in their paper, there cannot be more than one start or target position in $\mIS(F, F')$.
This avoids cycles of length 2.
Since $\W$ is simple, cycles of length 3 or more are also impossible.
Thus, if $\mu = \beta = 4$, $G$ is a DAG and a topological ordering can be found that respects interference between components.

However, with a tighter separation of $\beta = 3$, the claim that the mutual interference set $\mIS(F, F')$ can contain at most one single start or target position is no longer valid.
Since $\mu$ remains 4, it is still true that the mutual interference set cannot hold two or more start positions or two or more target positions.
However, it does allow the mutual interference set to contain both a start \emph{and} a target position.
See Figure~\ref{fig:example_loop} for an example of this interference.
If both start and target positions in the mutual interference set belong to the same free space component, the directed graph $G$ contains a cycle of length 2.
This breaks the topological ordering.

\begin{figure}[t]
    \centering
    \includegraphics{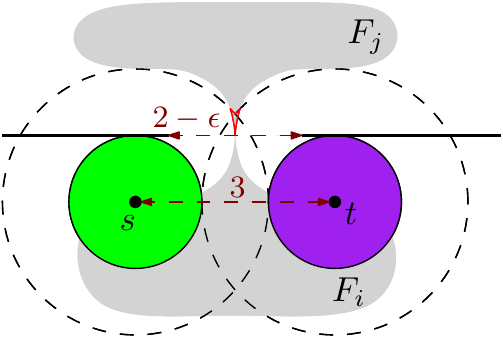}
    \caption{Example of start and target position both interfering with another free space component.}
    \label{fig:example_loop}
\end{figure}

Thus, for $\beta = 3$ the ordering breaks for the case where a start and a target positions in one free space component both interfere with another component.
In this case no direct ordering can be made since the component interferes when its start positions are occupied as well as when its target positions are occupied.

However, interference does not always affect the connectivity of the affected free space component.
Therefore, we define a \emph{remote blocker} as a target \emph{or} start position $x$ for which $\D_2(x)$ intersects the boundary of a free space component, other than the one it resides in, in more than one connected components.
By definition, all remote blockers between two free space components are in the mutual interference set.

\begin{restatable}{lemma}{noRemoteBlockers}
\label{lem:no-remote-blockers}
If the unlabeled motion planning problem has no remote blockers, then there is always a solution.
\end{restatable}
\begin{proof}
%We can create a motion plan for each free space component separately, using (one of) the algorithm(s) in Section~\ref{sec:single-component}.
We can create a motion plan for each free space component separately, using the algorithm in Section~\ref{sec:single-component}.

Let $x$ be an element of the interference set $\IS(F, F')$ between free space components $F$ and~$F'$.
Given the separation bounds $\mu = 4$ and $\beta = 3$, the aura of $x$ cannot contain another start/target position, thus any robot path $\pi$ that intersects the aura of $x$ does so in at least two points.
In other words, any path through the aura of $x$ must also leave the aura.
Since $x$ is not a remote blocker, its aura intersects the boundary of $F'$ in one connected component.
Therefore, every path $\pi$ that crosses the aura of $x$ can be modified to use $\partial \D_2(x)$ instead.
If there is another element in $\IS(F, F')$, then its aura and the aura of $x$ together intersect  a connected part of $\partial F'$, since $x, x'$ are both in $F$. After modifying the paths for each element in the interference sets we will arrive at a valid solution.
\end{proof}\begin{proof}
%We can create a motion plan for each free space component separately, using (one of) the algorithm(s) in Section~\ref{sec:single-component}.
We can create a motion plan for each free space component separately, using the algorithm in Section~\ref{sec:single-component}.

Let $x$ be an element of the interference set $\IS(F, F')$ between free space components $F$ and~$F'$.
Given the separation bounds $\mu = 4$ and $\beta = 3$, the aura of $x$ cannot contain another start/target position, thus any robot path $\pi$ that intersects the aura of $x$ does so in at least two points.
In other words, any path through the aura of $x$ must also leave the aura.
Since $x$ is not a remote blocker, its aura intersects the boundary of $F'$ in one connected component.
Therefore, every path $\pi$ that crosses the aura of $x$ can be modified to use $\partial \D_2(x)$ instead.
If there is another element in $\IS(F, F')$, then its aura and the aura of $x$ together intersect  a connected part of $\partial F'$, since $x, x'$ are both in $F$. After modifying the paths for each element in the interference sets we will arrive at a valid solution.
\end{proof}
Now the key geometric observation, that we prove in the lemma below, is that if auras of both $s$ and $t$ intersect $F'$, they cannot be both remote blockers of $F'$, and as a consequence we can still always find an order to resolve $F$ and $F'$.
%
%
%Identifying the remote blockers in the mutual interference set shows us which positions can break the connectivity and therefore potentially the solvability of the destination free space component.
%Looking at remote blockers, we can again try to find an ordering to the components that will satisfy the direction of the blockings.

\begin{lemma}
\label{lem:remote-blockers}
The interference set $\IS(F, F')$ of free space components $F, F'$ cannot contain a start position $x_1$ and target position $x_2$ that are both remote blockers of $F'$.
\end{lemma}
\begin{proof}
Let $F$ be a free space component containing a start position $x_1$ and a target position $x_2$ and let $F'$ be another free space component such that $x_1, x_2 \in \IS(F, F')$.
For a contradiction, we assume that both positions are remote blockers for $F'$.
Without loss of generality, assume that $x_1$ and $x_2$ lie on a horizontal line, and that $x_1$ is to the left of $x_2$.
%From Lemma 3 from the paper by Adler et al.~\cite{adler2015efficient}, we know that $\D_2(x_1) \cap \D_2(x_2) \neq \emptyset$.
%Together with $\beta = 3$, we conclude that $3 \leq d(x_1, x_2) < 4$, where $d(\cdot, \cdot)$ is the Euclidean distance function.
By our separation assumption $d(x_1,x_2)\geq 3$, where $d(\cdot, \cdot)$ is the Euclidean distance function.

As $x_1$ and $x_2$ are both remote blockers, their auras are split by the free space component $F'$ into multiple connected components.
Consider the components of $D_2(x_1)\setminus F'$ that do not contain $x_1$.
Let $y_1$ be the leftmost boundary point of $F'$ that lies in one of the components $D_2(x_1)\setminus F'$ not containing $x_1$.
Similarly, define $y_2$ to be the rightmost point on the boundary of $F'$ lying in one of the components $D_2(x_2)\setminus F'$ not containing $x_2$.
It then follows that the line segments $\overline{x_1 y_1}$ and $\overline{x_2 y_2}$ do not intersect.
%We take $y_1 \in \partial \D_2(x_1) \cap \partial F'$ and $y_2 \in \partial \D_2(x_2) \cap \partial F'$ such that the $d(y_1, y_2)$ is maximized.
%We assume without loss of generality that $x_1$ lies to the left of $x_2$ and that the line segments $\overline{x_1 y_1}$ and $\overline{x_2 y_2}$ do not intersect (in such a case we could choose $y_1$ and $y_2$ differently such that their distance does not decrease and the line segments do not intersect).
See Figure~\ref{fig:proof-remote-blockers} for an illustration.

\begin{figure}[b!]
    \centering
    \includegraphics[page=2]{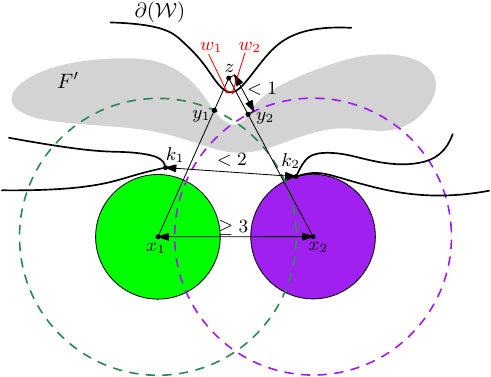} %app_prooFntersection_loop.pdf}
    \caption{Illustration of Lemma~\ref{lem:remote-blockers}. The aura $\D_2(x_1)$ is shown in green, $\D_2(x_2)$ in purple. 
    %The disks $K_1, K_2$ are shown in red. For simplicity we have shown when $x_1 = x_1'$ and $x_2 = x_2'$.
    }
    \label{fig:proof-remote-blockers}
\end{figure}

Since $d(x_1,y_1)=2$, and the distance from $x_1$ and $y_1$ to any point on the boundary of the workspace $\partial\W$ is at least $1$, we know that the straight-line segment $\overline{x_1 y_1}$ lies completely in $\W$.
Similarly $\overline{x_2 y_2} \subset \W$.
%Given that $x_1, x_2 \in F$, there exists a path $\pi_1 \subset F$ between $x_1$ and $x_2$.
%Similarly, there exists a path $\pi_2 \subset F'$ between $y_1$ and $y_2$.
%We now define the curve $\lambda =  \pi_1 \cup \overline{x_1 y_1} \cup \pi_2 \cup \overline{x_2 y_2}$  for which we can conclude that $\lambda \subset \W$.
%Let $A$ be the area enclosed by $\lambda$.
%Since $\W$ is simple and $\lambda \in \W$, we have $A \subset \W$.
%
%We now define points $x_1', y_1'$ to be points on $\overline{x_1 y_1}$ such that $x_1' \in F$ and $y_1' \in F'$ and the distance $d(x_1', y_1')$ is minimized.
%We do similar for $x_2'$ and $y_2'$ for the line segment $\overline{x_2 y_2}$.
%Take a unit-disc $K_1$ such that $K_1$ lies to the left of $\overline{x_1 y_1}$ and passes through $x_1'$ and $y_1'$.
%Similarly, take a unit-disc $K_2$ for $x_2'$ and $y_2'$ that lies to the right of $\overline{x_2 y_2}$.
%We define $A^* = A \setminus \F$.
%Then we have that $A^* \subset K_1 \cup K_2$.
%From this we can conclude that $K_1 \cap K_2 \neq \emptyset$, otherwise $F$ and $F'$ are connected.
As $F$ and $F'$ are separate connected components of the free space, the passage on the workspace between them is narrower than $2$.
Consider the closest pair of points $k_1$ and $k_2$ on $\partial\W$ on the opposite sides of this passage (see Figure~\ref{fig:proof-remote-blockers}).
The distance between them $d(k_1, k_2)$ is smaller than $2$.
Note that as $\overline{x_1y_1}\in\W$ and $\overline{x_2y_2}\in\W$, the straight-line segment $\overline{k_1k_2}$ intersects both $\overline{x_1y_1}$ and $\overline{x_2y_2}$.

Since $d(x_1, x_2) \geq 3$ and $d(k_1, k_2) < 2$, we can conclude that $d(y_1, y_2) < 1$.
Let $z$ be the intersection of the extensions of $\overline{x_1y_1}$ and $\overline{x_2y_2}$ beyond $y_1$ and $y_2$ respectively.
Then, as $d(x_1,x_2)\geq 3$, we get that $d(y_1,z) < 1$ and $d(y_2, z) < 1$.
Thus $z$ lies in the workspace $\W$.
% Needs to be elaborated

Now, consider point $w_1$ on the boundary of $\W$ at distance $1$ from $y_1$, that is $w_1\in \D_1(y_1)\cap\partial \W$, and consider point $w_2\in \D_1(y_2)\cap\partial \W$.
First, observe, that as $y_1$ lies to the left (along the boundary of $F'$) of $y_2$, the point $w_1$ lies to the left along the boundary of $\W$.
By the choice of $y_1$ and $y_2$ we have that $w_1$ must lie to the right of the line passing through $x_1$ and $y_1$, and that $w_2$ must lie to the left of the line passing through $x_2$ and $y_2$.
Otherwise, $y_1$ (or $y_2$) would not be the leftmost (rightmost) boundary point of $F'$.
%By how we defined $y_1$, we know there exists $w_1 \in \mathcal{O} \cap \D_1(y_1)$ and that $w_1$ lies on or to the right of the line $x_1y_1$ (otherwise there exists a point $y_1'' \in \partial \D_2(x_1) \cap \partial F'$ to the left of $y_1$ which would mean $d(y_1, y_2) < d(y_1'', y_2)$).
%Similarly, for $y_2$ there exists $w_2 \in \mathcal{O} \cap \D_1(y_2)$ which lies to the left of line segment $\overline{x_2y_2}$.
However, this contradicts the fact that $z$ lies inside the workspace.
Thus, $x_1$ and $x_2$ cannot be both remote blockers at the same time.
\end{proof}

Since the interference set cannot contain more than one remote blocker by Lemma~\ref{lem:remote-blockers}, we can find a topological ordering of solving the free space components such that the motion plan for a free space component $F$ is never blocked by a robot in a component $F' \neq F$.
Intuitively, since we cannot have both a start and target position as a remote blocker from one free space component to the other, there is a well-defined direction regarding remote blockers between two components.

Note that a position that is not a remote blocker can still interfere. However, any path in the motion plan that crosses the aura of an interfering position can be modified to use the boundary of the aura instead, see Lemma~\ref{lem:no-remote-blockers}. 

%\danny{I'd start the next theorem with ``We are give'', as otherwise the first sentence is incomplete.}
\begin{restatable}{theorem}{finalthm}
We are given $m$ unit-disc robots in a simple polygonal workspace $\W \subset \R^2$, with start and target positions $S, T$ and separation constraints $\mu = 4$ and $\beta = 3$.
Assuming each connected component $F$ of the free space $\F$ for a single unit-disc robot in $\W$ contains an equal number of start and target positions, there exists a collision-free motion plan for the robots starting at $S$ such that all target positions in $T$ are occupied after execution.
\end{restatable}
\begin{proof}
We can create a motion plan for each connected component $F \subset \F$ of the free space using 
%one of 
the algorithm discussed in Section~\ref{sec:single-component} %(Theorem~\ref{th:matching_correctness}, 
(Theorem~\ref{th:divide_conquer_correctness}).
We can then use the fact that only a single position can be a remote blocker between two components of $\F$ by Lemma~\ref{lem:remote-blockers} to find an ordering for solving the free space components that respects remote blockers.

To obtain the ordering of free space components, a directed forest $G = (V, E)$ is created such that the nodes in $V$ correspond to the free space components $F \in \F$.
We add the directed edge $(F, F')$ to $E$ if either there exists a start position $s \in F$ such that $s$ is a remote blocker for $F'$, or there exists a remote blocker target position $t \in F'$ such that $t$ is a remote blocker for $F$.

Given the fact that an interference set cannot contain two remote blockers by Lemma~\ref{lem:remote-blockers} and the workspace $\W$ is simple, the graph $G$ is a DAG.
Therefore, $G$ has a topological ordering that respects remote blocking between components.

Lastly, the motion plan of one free space component might still encounter interference from other free space components.
But since these are not remote components, we are able to modify all paths that pass the aura of an interfering position to take the boundary instead, as explained in Lemma~\ref{lem:no-remote-blockers}.
\end{proof}

% ================================================
% =================== Conclusion =================
% ================================================

\section{Conclusion}
\label{sec:conclusion}

In this paper we presented an efficient algorithm for the unlabeled motion planning problem for unit-disc robots with sufficient separation in a simple polygon. Our result is optimal, in the sense that with less separation a solution may not exist. Nevertheless, there remain a number of challenging open problems. 

To prove the tightness of the separation bounds, we first constructed domains with straight-line segments and circular arcs as boundaries, and then obtained simple polygons by approximating these. This results in polygons of high complexity. An open question remains whether it is possible to prove the separation bounds with constant-complexity polygons.

Of course, a solution may still exist even if the separation bounds are violated. The complexity of the problem in this setting remains a challenging open problem. The general unlabeled motion planning problem in a polygonal environment with holes is PSPACE-complete~\cite{brocken2020multi,solovey2016hardness}. Does the restriction to unit-disc robots and/or simple domains make the problem tractable, in particular if we still enforce some small separation bound?

What challenges arise when the workspace is no longer simple, but rather contains obstacles?
Intuitively, obstacles seem to pose an issue when defining an ordering for solving multiple free space components, since positions can interfere between components at multiple locations. Are there conditions, similar to the separation bounds, which can guarantee a solution (together with an efficient algorithm) for unlabeled multi-robot motion planning amidst obstacles?

\subparagraph{Acknowledgments}
This research was initiated at the Lorentz-Center Workshop on Fixed-Parameter Computational Geometry.
Mark de Berg is supported by the  Dutch Research Council (NWO) through Gravitation project NETWORKS  (grant no. 024.002.003).
Dan Halperin is supported in part by the Israel Science Foundation (grant no.~1736/19), by NSF/US-Israel-BSF (grant no.~2019754), by the Israel Ministry of Science and Technology (grant no.~103129), by the Blavatnik Computer Science Research Fund, and by the Yandex Machine Learning Initiative for Machine Learning at Tel Aviv University.
Yoshio Okamoto is supported by the JSPS KAKENHI Grant Numbers JP20H05795 and JP20K11670.

\bibliography{citations}

\begin{thebibliography}{10}

\bibitem{adler2015efficient}
Aviv Adler, Mark de~Berg, Dan Halperin, and Kiril Solovey.
\newblock Efficient multi-robot motion planning for unlabeled discs in simple
  polygons.
\newblock In {\em Algorithmic Foundations of Robotics XI}, pages 1--17.
  Springer, 2015.

\bibitem{brocken2020multi}
Thomas Brocken, G.~Wessel van~der Heijden, Irina Kostitsyna, Lloyd~E. Lo-Wong,
  and Remco J.~A. Surtel.
\newblock Multi-robot motion planning of $k$-colored discs is {PSPACE}-hard.
\newblock In {\em 10th International Conference on Fun with Algorithms (FUN
  2021)}, volume 157 of {\em Leibniz International Proceedings in Informatics
  (LIPIcs)}, pages 15:1--15:16, 2020.

\bibitem{choset2005principles}
Howie Choset, Kevin~M. Lynch, Seth Hutchinson, George Kantor, Wolfram Burgard,
  Lydia~E. Kavraki, and Sebastian Thrun.
\newblock {\em Principles of Robot Motion: Theory, Algorithms, and
  Implementation}.
\newblock MIT Press, 2005.

\bibitem{demaine2019coordinated}
Erik~D. Demaine, S{\'{a}}ndor~P. Fekete, Phillip Keldenich, Henk Meijer, and
  Christian Scheffer.
\newblock Coordinated motion planning: Reconfiguring a swarm of labeled robots
  with bounded stretch.
\newblock {\em {SIAM} Journal on Computing}, 48(6):1727--1762, 2019.

\bibitem{dudley-approx}
Richard~M. Dudley.
\newblock Metric entropy of some classes of sets with differentiable
  boundaries.
\newblock {\em Journal of Approximation Theory}, 10(3):227--236, 1974.

\bibitem{halperin2018robotics}
Dan Halperin, Lydia Kavraki, and Kiril Solovey.
\newblock Robotics.
\newblock In Jacob~E. Goodman, Joseph O'Rourke, and Csaba T\'oth, editors, {\em
  Handbook of Discrete and Computational Geometry}, chapter~51, pages
  1343--1376. Chapman \& Hall/CRC, 3rd edition, 2018.

\bibitem{halperin2018algorithmic}
Dan Halperin, Micha Sharir, and Oren Salzman.
\newblock Algorithmic motion planning.
\newblock In Jacob~E. Goodman, Joseph O'Rourke, and Csaba T\'oth, editors, {\em
  Handbook of Discrete and Computational Geometry}, chapter~50, pages
  1311--1342. Chapman \& Hall/CRC, 3rd edition, 2018.

\bibitem{hearn2005pspace}
Robert~A. Hearn and Erik~D. Demaine.
\newblock {PSPACE}-completeness of sliding-block puzzles and other problems
  through the nondeterministic constraint logic model of computation.
\newblock {\em Theoretical Computer Science}, 343:72--96, 2005.

\bibitem{hopcroft1984complexity}
John~E. Hopcroft, Jacob~Theodore Schwartz, and Micha Sharir.
\newblock On the complexity of motion planning for multiple independent
  objects; {PSPACE}-hardness of the ``warehouseman's problem''.
\newblock {\em The International Journal of Robotics Research}, 3(4):76--88,
  1984.

\bibitem{kavraki1996probabilistic}
Lydia~E. Kavraki, Petr Svestka, Jean-Claude Latombe, and Mark~H. Overmars.
\newblock Probabilistic roadmaps for path planning in high-dimensional
  configuration spaces.
\newblock {\em IEEE Transactions on Robotics and Automation}, 12(4):566--580,
  1996.

\bibitem{kedem1986union}
Klara Kedem, Ron Livne, J{\'a}nos Pach, and Micha Sharir.
\newblock On the union of {J}ordan regions and collision-free translational
  motion amidst polygonal obstacles.
\newblock {\em Discrete \& Computational Geometry}, 1(1):59--71, 1986.

\bibitem{kloder2006path}
Stephen Kloder and Seth Hutchinson.
\newblock Path planning for permutation-invariant multirobot formations.
\newblock {\em IEEE Transactions on Robotics}, 22(4):650--665, 2006.

\bibitem{kornhauser1984coordinating}
Daniel~M. Kornhauser, Gary Miller, and Paul Spirakis.
\newblock Coordinating pebble motion on graphs, the diameter of permutation
  groups, and applications.
\newblock Master's thesis, MIT, Dept. of Electrical Engineering and Computer
  Science, 1984.

\bibitem{kuffner2000rrtcpnnect}
James~J. Kuffner and Steven~M. Lavalle.
\newblock {RRT}-{C}onnect: An efficient approach to single-query path planning.
\newblock In {\em IEEE International Conference on Robotics and Automation
  (ICRA)}, pages 995--1001, 2000.

\bibitem{lavalle2006planning}
Steven~M. LaValle.
\newblock {\em Planning Algorithms}.
\newblock Cambridge University Press, 2006.

\bibitem{sanchez2002using}
Gildardo Sanchez and Jean-Claude Latombe.
\newblock Using a {PRM} planner to compare centralized and decoupled planning
  for multi-robot systems.
\newblock In {\em IEEE International Conference on Robotics and Automation},
  volume~2, pages 2112--2119. IEEE, 2002.

\bibitem{schwartz1983piano}
Jacob~T. Schwartz and Micha Sharir.
\newblock On the ``piano movers'' problem. {II}. {G}eneral techniques for
  computing topological properties of real algebraic manifolds.
\newblock {\em Advances in Applied Mathematics}, 4(3):298--351, 1983.

\bibitem{schwartz1983piano2}
Jacob~T. Schwartz and Micha Sharir.
\newblock On the piano movers' problem: {III}. coordinating the motion of
  several independent bodies: The special case of circular bodies moving amidst
  polygonal barriers.
\newblock {\em The International Journal of Robotics Research}, 2(3):46--75,
  1983.

\bibitem{sharir1991coordinated}
Micha Sharir and Shmuel Sifrony.
\newblock Coordinated motion planning for two independent robots.
\newblock {\em Annals of Mathematics and Artificial Intelligence},
  3(1):107--130, 1991.

\bibitem{shome2020drrtstar}
Rahul Shome, Kiril Solovey, Andrew Dobson, Dan Halperin, and Kostas~E. Bekris.
\newblock {dRRT\({}^{\mbox{*}}\)}: Scalable and informed asymptotically-optimal
  multi-robot motion planning.
\newblock {\em Autonomous Robots}, 44(3-4):443--467, 2020.

\bibitem{solomon2018motion}
Israela Solomon and Dan Halperin.
\newblock Motion planning for multiple unit-ball robots in {$\mathbb{R}^{d}$}.
\newblock In Marco Morales, Lydia Tapia, Gildardo S{\'{a}}nchez{-}Ante, and
  Seth Hutchinson, editors, {\em Proc. 13th Workshop on the Algorithmic
  Foundations of Robotics, {WAFR}}, volume~14 of {\em Springer Proceedings in
  Advanced Robotics}, pages 799--816. Springer, 2018.

\bibitem{solovey2016hardness}
Kiril Solovey and Dan Halperin.
\newblock On the hardness of unlabeled multi-robot motion planning.
\newblock {\em The International Journal of Robotics Research},
  35(14):1750--1759, 2016.

\bibitem{solovey2016finding}
Kiril Solovey, Oren Salzman, and Dan Halperin.
\newblock Finding a needle in an exponential haystack: Discrete {RRT} for
  exploration of implicit roadmaps in multi-robot motion planning.
\newblock {\em International Journal of Robotics Research}, 35(5):501--513,
  2016.

\bibitem{solovey2015motion}
Kiril Solovey, Jingjin Yu, Or~Zamir, and Dan Halperin.
\newblock Motion planning for unlabeled discs with optimality guarantees.
\newblock In {\em Robotics: Science and Systems XI}. Robotics: Science and
  Systems Foundation, 2015.

\bibitem{spirakis1984strong}
Paul Spirakis and Chee~K. Yap.
\newblock Strong {NP}-hardness of moving many discs.
\newblock {\em Information Processing Letters}, 19(1):55--59, 1984.

\bibitem{stern2019multiagent}
Roni Stern, Nathan~R. Sturtevant, Ariel Felner, Sven Koenig, Hang Ma, Thayne~T.
  Walker, Jiaoyang Li, Dor Atzmon, Liron Cohen, T.~K.~Satish Kumar, Roman
  Bart{\'{a}}k, and Eli Boyarski.
\newblock Multi-agent pathfinding: Definitions, variants, and benchmarks.
\newblock In Pavel Surynek and William Yeoh, editors, {\em Proc. 12th
  International Symposium on Combinatorial Search, {SOCS}}, pages 151--159.
  {AAAI} Press, 2019.

\bibitem{svestka1998coordinated}
Petr Svestka and Mark~H. Overmars.
\newblock Coordinated path planning for multiple robots.
\newblock {\em Robotics and Autonomous Systems}, 23(3):125--152, 1998.

\bibitem{turpin2013concurrent}
Matthew Turpin, Nathan Michael, and Vijay Kumar.
\newblock Concurrent assignment and planning of trajectories for large teams of
  interchangeable robots.
\newblock In {\em IEEE International Conference on Robotics and Automation},
  pages 842--848. IEEE, 2013.

\bibitem{wagner2015subdimensional}
Glenn Wagner and Howie Choset.
\newblock Subdimensional expansion for multirobot path planning.
\newblock {\em Artificial Intelligence}, 219:1--24, 2015.

\bibitem{yap1984coordinating}
Chee Yap.
\newblock Coordinating the motion of several discs.
\newblock {\em Courant Institute of Mathematical Sciences}, 1984.

\bibitem{yu2018constant}
Jingjin Yu.
\newblock Constant factor time optimal multi-robot routing on high-dimensional
  grids.
\newblock In Hadas Kress{-}Gazit, Siddhartha~S. Srinivasa, Tom Howard, and
  Nikolay Atanasov, editors, {\em Robotics: Science and Systems XIV}, 2018.

\bibitem{yu2016optimal}
Jingjin Yu and Steven~M. LaValle.
\newblock Optimal multirobot path planning on graphs: Complete algorithms and
  effective heuristics.
\newblock {\em {IEEE} Transactions on Robotics}, 32(5):1163--1177, 2016.

\end{thebibliography}

\end{document}